\documentclass[onecolumn,draft,11pt]{IEEEtran}%[conference]{IEEEtran}
\usepackage{dsfont,graphicx,color,amsfonts,epsf,epsfig,verbatim,amssymb,amsmath,array,cite,amsthm,subfigure,multicol,multirow}

\ifCLASSINFOpdf
\IEEEoverridecommandlockouts
\else
\fi

\begin{document}
\title{Abelian Group Codes for Source Coding and Channel Coding}

\author{\IEEEauthorblockN{Aria G. Sahebi and S. Sandeep
    Pradhan\thanks{This work was supported by NSF grant
      CCF-1116021. This work was presented in part at IEEE
      International Symposium on Information Theory (ISIT), July
      2011, and Allerton conference on communication, conrol and
      computing, October 2012.  }}\\
\IEEEauthorblockA{Department of Electrical Engineering and Computer Science,\\
University of Michigan, Ann Arbor, MI 48109, USA.\\
Email: \tt\small ariaghs@umich.edu, pradhanv@umich.edu}}

\newtheorem{theorem}{Theorem}[section]
\newtheorem{deff}{Definition}[section]
\newtheorem{example}{Example}[section]
\newtheorem{lemma}[theorem]{Lemma}
\newtheorem{prop}[theorem]{Proposition}
\newtheorem{cor}[theorem]{Corollary}
\newtheorem{remark}[theorem]{Remark}
\maketitle

%\begin{abstract}
%%\boldmath
%\end{abstract}

\IEEEpeerreviewmaketitle
\begin{abstract}
In this paper, we study the asymptotic performance of Abelian group codes for the
lossy source coding problem for arbitrary discrete (finite alphabet) memoryless sources
as well as the channel coding problem for arbitrary discrete (finite alphabet)
memoryless channels. For the source coding problem, we derive an
achievable rate-distortion function that is characterized in a
single-letter information-theoretic form  using the ensemble of Abelian
group codes. When the underlying group is a field, it
simplifies to the symmetric rate-distortion function. Similarly, for
the channel coding problem, we find an  achievable rate characterized
in a single-letter information-theoretic form using group
codes. This simplifies to the symmetric capacity of the
channel when the underlying group is a field. We compute the
rate-distortion function and the achievable rate for several examples
of sources and channels. Due to the non-symmetric nature of  the
sources and channels considered, our analysis uses a synergy of information
theoretic and group-theoretic tools.

\end{abstract}

\section{Introduction}
Approaching information theoretic performance limits of communication
using structured codes has been of great interest for the last several decades
\cite{ahlswede_alg_codes,Goblick,Dobrusin,dobrushin_group}.
The earlier attempts to design computationally efficient  encoding and decoding algorithms
for point-to-point communication (both channel coding and source
coding) resulted in injection of finite field structures to the coding
schemes \cite{CSBook}. In the channel coding problem \cite{MacBook},
the channel input alphabets are replaced with algebraic fields and
encoders are replaced with  matrices. Similarly in source
coding problem \cite{CoveringBook}, the reconstruction alphabets are
replaced with a finite fields and decoders are replaced with matrices. Later these coding approaches were extended to
weaker algebraic structures such as rings and groups
\cite{forney_dynamics,fagnani_abelian,Loeliger96,Garello95,loeliger-signal}\footnote{Note
  that this is an incomplete list. There is a vast body of work on
  group codes. See \cite{CSBook} for a more complete bibliography.}.
The motivation for this are two fold: a) finite fields exist only for alphabets
with size equal to a prime power, and b) for communication under certain constraints,
codes with weaker algebraic structures have better properties.
For example, when communicating over an additive white Gaussian noise
channel with $8$-PSK constellation, codes over $\mathds{Z}_8$, the
cyclic group of size $8$, are more desirable over binary linear codes
because the structure of the code is matched to the structure of the signal set
\cite{loeliger-signal}, and hence the former have superior error
correcting properties. As another example, construction of polar codes
over alphabets of size $p^r$,  for $r>1$ and $p$ prime,  is simpler with a module structure
rather than a vector space structure \cite{Sahebi_polar_allerton2011,sasoglu_polar_q,Park_Barg_Polar}.
Subsequently, as interest in network information theory grew, these codes were
used to approach the information-theoretic  performance limits of certain special cases of multi-terminal
communication problems \cite{Wyner74,Zamir02,sandeep_discus,Erez05}. These limits were
obtained earlier using the random coding ensembles in the information theory
literature.

In 1979,   Korner and Marton, in a significant departure from
tradition,  showed that for
a binary distributed source coding problem, the asymptotic  average performance of
binary linear code ensembles can be superior to that of the standard
random coding ensembles.  Although, structured codes were being used in communication
mainly for computational complexity reasons, the duo showed that, in
contrast, even when computational complexity is a non-issue, the use of structured
codes leads to superior asymptotic performance limits in multi-terminal communication problems. In the recent
past, such gains were shown for a wide class of
problems \cite{phiosof_zamir_journal,dinesh_dsc,nazer_gastpar,vinodh_group_source_coding,Sridharan10}.
In our prior work, we developed an inner bound
to the optimal rate-distortion region for the distributed source
coding problem  \cite{dinesh_dsc,sahebi_DSC_Allerton} in which Abelian group codes were used as building blocks in the
coding schemes. Similar coding approaches were applied for the
interference channel and the broadcast channel in
\cite{padakandla_3_user_inter,padakandla_broadcast_arxiv}.
The motivation for studying Abelian group codes beyond the non-existence of finite fields
over arbitrary alphabets  is the following.
The algebraic structure of the code imposes certain restrictions on the performance.
For certain problems, linear codes were shown to be not
optimal\cite{dinesh_dsc}, and  Abelian
group codes exhibit a superior performance. For example, consider a
distributed source coding problem with two statistically correlated
but individually uniform quaternary sources $X$ and $Y$ that are related via the relation $X=Y+Z$, where $+$ denotes addition modulo-$4$
and $Z$ is a hidden quaternary random variable that has a non-uniform
distribution and is independent of $Y$. The joint decoder wishes to
reconstruct $Z$ losslessly. In this problem, codes over $\mathds{Z}_4$ perform better than linear codes over
the Galois field of size $4$.
In summary, the main reason for using algebraic structured codes in this context is performance
rather than complexity of encoding and decoding.
Hence information-theoretic characterizations of asymoptotic
performance of Abelian group code ensembles for various communication problems and under various
decoding constraints became important.

Such performance limits have been characterized in certain special
cases.  It is well-known that binary linear codes achieve the capacity of
binary symmetric channels \cite{elias}. More generally, it has also been
shown that $q$-ary linear codes can achieve the capacity of symmetric
channels \cite{dobrushin_group} and linear codes can be used to compress
a source losslessly down to its entropy \cite{korner_marton}.
Goblick \cite{Goblick} showed that binary linear codes achieve the rate-distortion function
of binary uniform sources with Hamming distortion criterion.
Group codes were first studied by Slepian \cite{slepian_group} for the Gaussian channel. In
\cite{ahlswede_group}, the capacity of group codes for certain classes of
channels has been computed. Further results on the capacity of group codes
were established in \cite{ahlswede_alg_codes,ahlswede_alg_codes2}. The
capacity of group codes over a class of channels exhibiting symmetries with
respect to the action of a finite Abelian group has been investigated in
\cite{fagnani_abelian}.

In this work, we focus on two problems. In the first, we consider the lossy source coding problem for arbitrary
discrete memoryless sources with single-letter distortion measures and
the reconstruction alphabet being  equipped with the structure of a finite
Abelian group $G$.  We derive an upper bound on the rate-distortion
function achievable using group codes which are subgroups of $G^n$,
where $n$ denotes the block length of encoding which is arbitrarily
large.  The average performance of the ensemble is shown
to be the symmetric rate-distortion function of the source when the underlying group
is a field i.e. the Shannon rate-distortion
function with the additional constraint that the reconstruction variable is uniformly distributed. For the general case, it
turns out that several additional terms appear corresponding to
subgroups of the underlying group in the form of a maximization and this can
result in a larger rate compared to the symmetric rate for a given distortion level.

In the second part, we consider the channel coding problem for
arbitrary discrete memoryless channels. We assume that the channel
input alphabet is equipped with the structure of a finite  Abelian group $G$.
 We derive a lower bound on the capacity of such channels achievable
 using group codes  which are subgroups of $G^n$. We show that the achievable rate is equal to the
symmetric capacity of the channel when the underlying group is a field; i.e. it is equal to the Shannon mutual information between the channel
input and the channel output when the channel input is uniformly
distributed. Similar to the source coding, we show that in the general
case, several
additional terms appear corresponding to subgroups of the underlying group in the form of a minimization and the achievable rate can be
smaller than the symmetric capacity of the channel.

It can be noted that the bounds on the performance limits as mentioned
above apply to any arbitrary discrete memoryless case. Moreover, we
use joint typicality encoding and decoding \cite{Csiszarbook} for both problems at
hand. This will make the analysis more tractable. In this approach we
use a synergy of information-theoretic and group-theoretic  tools.
The traditional approaches have looked at
encoding and decoding of structured codes based on  either minimum distance or maximum likelihood.
We introduce two information quantities that capture the performance
limits achievable using Abelian group codes that are analogous to the mutual
information which captures the Shannon performance limits when no
algebraic structure is enforced on the codes. They are source coding
group mutual information and channel coding group mutual
information. The converse bounds for both problems will be addressed in a future correspondence.

The paper is organized as follows: In section \ref{Preliminaries}, some definitions and basic facts are stated which are
used in the paper. In Section \ref{section:Abelian}, we introduce the ensemble of Abelian group codes used in the paper.
In section \ref{section:main_results}, we state the main results of the paper for both the source coding problem as well as the channel coding problem. We also simplify the expressions for the case where the underlying group is a
$\mathds{Z}_{p^r}$ ring. In Section \ref{section:proof_source}, we prove the results for the source coding problem and similarly, in Section \ref{section:proof_channel}, we prove the results for the channel coding problem. In Section \ref{section:Examples}, we show that for the source coding problem, when the underlying group is a field, the
rate-distortion function achievable using Abelian group codes is equal to the symmetric rate-distortion function of the source
and for the channel coding problem, the rate achievable using Abelian group codes is equal to the symmetric capacity of the channel. We also provide
several examples dealing with non-field groups in this section.  We conclude in
Section \ref{section:Conclusion}.

\section{Preliminaries} \label{Preliminaries}

\subsubsection{Source Model}
The source is modeled as a discrete-time memoryless random process $X$ with each sample
taking values from a finite set $\mathcal{X}$ called alphabet according to the distribution $p_{X}$. The reconstruction alphabet is denoted by $\mathcal{U}$ and the quality of reconstruction is measured by a single-letter distortion functions $d:\mathcal{X}\times \mathcal{U}\rightarrow \mathds{R}^{+}$. We denote this source by $(\mathcal{X},\mathcal{U},p_{X},d)$.

\subsubsection{Channel Model}
We consider discrete memoryless channels used without feedback. We associate two finite sets $\mathcal{X}$ and $\mathcal{Y}$ with the channel as the channel input and output alphabets. The input-output relation of the channel is characterized by a conditional probability law $W_{Y|X}(y|x)$ for $x\in \mathcal{X}$ and $y\in \mathcal{Y}$. The channel is specified by $(\mathcal{X},\mathcal{Y},W_{Y|X})$.

\subsubsection{Groups}
All groups referred to in this paper are \emph{Abelian groups}. Given a group $(G,+)$, a subset $H$ of $G$ is called a \emph{subgroup} of $G$ if it is closed under the group operation. In this case, $(H,+)$ is a group in its own right. This is denoted by $H\le G$. A \emph{coset} $C$ of a subgroup $H$ is a shift of $H$ by an arbitrary element $a\in G$ (i.e. $C=a+H$ for some $a\in G$). For a subgroup $H$ of $G$, the number of cosets of $H$ in $G$ is called the \emph{index} of $H$ in $G$ and is denoted by $|G:H|$. The index of $H$ in $G$ is equal to $|G|/|H|$ where $|G|$ and $|H|$ are the cardinality or size of $G$ and $H$ respectively. For a prime $p$ dividing the cardinality of $G$, the \emph{Sylow}-$p$ subgroup of $G$ is the largest subgroup of $G$ whose cardinality is a power of $p$.
Group isomorphism is denoted by $\cong$. \\

\subsubsection{Group Codes}
Given a group $G$, a group code $\mathds{C}$ over $G$ with block length $n$ is any subgroup of $G^n$. A shifted group code over $G$, $\mathds{C}+B$ is a translation of a group code $\mathds{C}$ by a fixed vector $B\in G^n$. Group codes generalize the notion of linear codes over {fields} to sources with reconstruction alphabets (and channels with input alphabets) having composite sizes.\\

\subsubsection{Achievability for Source Coding and the Rate-Distortion Function}
For a group $G$, a group transmission system with parameters $(n,\Theta,\Delta,\tau)$ for compressing a given source $(\mathcal{X},\mathcal{U}=G,P_{X},d)$ consists of a codebook, an encoding mapping and a decoding mapping. The codebook $\mathds{C}$ is a shifted subgroup of $G^n$ whose size is equal to $\Theta$ and the mappings are defined as %$e:\mathcal{X}^n\rightarrow \{1,2,\cdots,\Theta\}$ and a decoding mapping $g:\mathcal{S}^n\times \{1,2,\cdots,\Theta\}\rightarrow \mathcal{U}^n$
\begin{align*}
&\mbox{\small Enc}:\mathcal{X}^n\rightarrow \{1,2,\ldots,\Theta \},\\
&\mbox{\small Dec}:\{1,2,\ldots,\Theta \} \rightarrow \mathds{C}
\end{align*}
such that% the following condition is met:%$P\left(d(X^n,g(e(X^n)))>\Delta\right)\le \tau$ %
\begin{align*}
P\left[d(X^n,\mbox{\small Dec}(\mbox{\small Enc}(X^n)))>\Delta\right]\le \tau
\end{align*}
where $X^n$ is the random vector of length $n$ generated by the source. In this transmission system, $n$ denotes the block length, $\log \Theta$ denotes the number of ``channel uses'', $\Delta$ denotes the distortion level and $\tau$ is a bound on the probability of exceeding the distortion level $\Delta$.\\
Given a source $(\mathcal{X},\mathcal{U}=G,P_{X},d)$, a pair of non-negative real numbers $(R,D)$ is said to be achievable using group codes if for every $\epsilon>0$ and for all sufficiently large numbers $n$, there exists a group transmission system with parameters $(n,\Theta,\Delta,\tau)$ for compressing the source such that
\begin{align*}
\frac{1}{n}\log \Theta\le R+\epsilon, \qquad \Delta\le D+\epsilon,\qquad \tau\le \epsilon
\end{align*}
The optimal group rate-distortion function $R^*(D)$ of the source is given by the infimum of the rates $R$ such that $(R,D)$ is achievable using group codes.\\

\subsubsection{Achievability for Channel Coding}
For a group $G$, a group transmission system with parameters $(n,\Theta,\tau)$ for reliable communication over a given channel $(\mathcal{X}=G,\mathcal{Y},W_{Y|X})$ consists of a codebook, an encoding mapping and a decoding mapping. The codebook $\mathds{C}$ is a shifted subgroup of $G^n$ whose size is equal to $\Theta$ and the mappings are defined as
\begin{align*}
&\mbox{\small Enc}:\{1,2,\cdots,\Theta\}\rightarrow \mathds{C}\\
&\mbox{\small Dec}:\mathcal{Y}^n\rightarrow\{1,2,\cdots,\Theta\}
\end{align*}
such that
\begin{align*}
%\mathds{E}_{p_S}\left\{\sum_{m=1}^{\Theta}\frac{1}{\Theta}P\left(\mbox{Dec}(Y^n)\ne m|X^n=\mbox{Enc}(S^n,m)\right)\right\}\le \tau
%\int_{s\in\mathcal{S}^n}f_S^n(s)\int_{\substack{y\in\mathcal{Y}^n\\f(y)\ne m}}f_{Y|XS}^n\left(y|e(m),s \right) ds dy\le \tau
\sum_{m=1}^\Theta \frac{1}{\Theta} \sum_{x\in\mathcal{X}^n} \mathds{1}_{\{x=\mbox{\scriptsize Enc}(m)\}} \sum_{y\in\mathcal{Y}^n} \mathds{1}_{\{m\ne \mbox{\scriptsize Dec}(y)\}} W^n(y|x)\le \tau
\end{align*}
Given a channel $(\mathcal{X}=G,\mathcal{Y},W_{Y|X})$, the rate $R$ is said to be achievable using group codes if for all $\epsilon>0$ and for all sufficiently large $n$, there exists a group transmission system for reliable communication with parameters $(n,\Theta,\tau)$ such that
\begin{align*}
\frac{1}{n}\log \Theta \ge R-\epsilon,\qquad\tau\le \epsilon
\end{align*}
The group capacity of the channel $C$ is defined as the supremum of the set of all achievable rates using group codes.\\

\subsubsection{Typicality}
Consider two random variables $X$ and $Y$ with joint probability mass
 function $p_{X,Y}(x,y)$ over $\mathcal{X}\times\mathcal{Y}$. Let $n$ be an integer and $\epsilon$ be a positive real number. The sequence pair $(x^n,y^n)$ belonging to $\mathcal{X}^n\times \mathcal{Y}^n$ is said to be jointly $\epsilon$-typical with respect to $p_{X,Y}(x,y)$ if
\begin{align*}
\nonumber \forall a\in\mathcal{X},\mbox{   }\forall b\in\mathcal{Y}:
\left|\frac{1}{n}N\left(a,b|x^n,y^n\right)-p_{X,Y}(a,b)\right|\le
\frac{\epsilon}{|\mathcal{X}||\mathcal{Y}|}
\end{align*}
and none of the pairs $(a,b)$ with $p_{X,Y}(a,b)=0$ occurs in $(x^n,y^n)$. Here, $N(a,b|x^n,y^n)$ counts the number of occurrences of the pair $(a,b)$ in the sequence pair $(x^n,y^n)$. We denote the set of all jointly $\epsilon$-typical sequence pairs in $\mathcal{X}^n\times \mathcal{Y}^n$ by $A_\epsilon^n(X,Y)$.\\
Given a sequence $x^n\in A_\epsilon^n(X)$, the set of conditionally $\epsilon$-typical sequences $A_\epsilon^n(Y|x^n)$ is defined as
\begin{align*}
A_\epsilon^n(Y|x^n)=\left\{y^n\in \mathcal{Y}^n\left| (x^n,y^n)\in
A_\epsilon^n(X,Y)\right.\right\}
\end{align*}

\subsubsection{Notation}
In our notation, $O(\epsilon)$ is any function of $\epsilon$ such that $\lim_{\epsilon\downarrow 0}O(\epsilon)=0$, $\mathds{P}$ is
the set of all primes, $\mathds{Z}^+$ is the set of positive integers and $\mathds{R}^+$ is the set of non-negative reals.
Since we deal with summations over several groups in this paper, when not clear from the context, we indicate the underlying group in each summation;
e.g. summation over the group $G$ is denoted by $\overbrace{\sum}^{(G)}$. Direct sum of groups is denoted by
$\bigoplus$ and direct product of sets is denoted by $\bigotimes$.

%%%%%%%%%%%%%%%%%%%%%%%%%%%%%%%%%%%%%%%%%%%%%%%%%%%%%%%%%%%%%%%%%%%%%%%%%%%%%%%%%%%%%%%%%%%%%%%%%%%%%%%%%%%%%%%%%%%%%%%

\section{Abelian Group Code Ensemble}\label{section:Abelian}
In this section, we use a standard characterization of Abelian groups and introduce the ensemble of Abelian group codes used in the paper.
\subsection{Abelian Groups}
For an Abelian group $G$, let $\mathcal{P}(G)$ denote the set of all distinct primes which divide $|G|$ and for a prime $p\in\mathcal{P}(G)$ let $S_p(G)$ be the corresponding Sylow subgroup of $G$. It is known \cite[Theorem 3.3.1]{group_hall} that any Abelian group $G$ can be decomposed as a direct sum of its Sylow subgroups in the following manner
\begin{align}\label{eqn:G_decomposition}
G=\bigoplus_{p\in \mathcal{P}(G)} S_p(G)
\end{align}
Furthermore, each Sylow subgroup $S_p(G)$ can be decomposed into $\mathds{Z}_{p^r}$ groups as follows:
\begin{align}\label{eqn:Sp_decomposition}
S_p(G)\cong \bigoplus_{r\in\mathcal{R}_p(G)} \mathds{Z}_{p^r}^{M_{p,r}}
\end{align}
where $\mathcal{R}_p(G)\subseteq \mathds{Z}^+$ and for $r\in\mathcal{R}_p(G)$, $M_{p,r}$ is a positive integer. Note that $\mathds{Z}_{p^r}^{M_{p,r}}$ is defined as the direct sum of the ring $\mathds{Z}_{p^r}$ with itself for $M_{p,r}$ times. Combining Equations \eqref{eqn:G_decomposition} and \eqref{eqn:Sp_decomposition}, we can represent any Abelian group as follows:
\begin{align}\label{eqn:G}
G\cong \bigoplus_{p\in\mathcal{P}(G)} \bigoplus_{r\in\mathcal{R}_p(G)} \mathds{Z}_{p^r}^{M_{p,r}}= \bigoplus_{p\in\mathcal{P}(G)} \bigoplus_{r\in\mathcal{R}_p(G)} \bigoplus_{m=1}^{M_{p,r}} \mathds{Z}_{p^r}^{(m)}
\end{align}
where $\mathds{Z}_{p^r}^{(m)}$ is called the $m\textsuperscript{th}$ $\mathds{Z}_{p^r}$ ring of $G$ or the $(p,r,m)\textsuperscript{th}$ ring of $G$. Equivalently, this can be written as follows
\begin{align*}
G\cong \bigoplus_{(p,r,m)\in\mathcal{G}(G)} \mathds{Z}_{p^r}^{(m)}
\end{align*}
where $\mathcal{G}(G)\subseteq \mathds{P}\times \mathds{Z}^+\times\mathds{Z}^+$ is defined as:
\begin{align*}
\mathcal{G}(G)=\{(p,r,m)\in\mathds{P}\times \mathds{Z}^+\times\mathds{Z}^+| p\in\mathcal{P}(G), r\in\mathcal{R}_p(G), m\in\{1,2,\cdots,M_{p,r}\}  \}
\end{align*}

This means any element $a$ of the Abelian group can be regarded as a vector whose components are indexed by $(p,r,m)\in\mathcal{G}(G)$ and whose $(p,r,m)\textsuperscript{th}$ component $a_{p,r,m}$ takes values from the ring $\mathds{Z}_{p^r}$. With a slight abuse of notation, we represent an element $a$ of $G$ as
\begin{align*}
a=\bigoplus_{(p,r,m)\in\mathcal{G}(G)} a_{p,r,m}
\end{align*}
Furthermore, for two elements $a,b\in G$, we have
\begin{align*}
a+b=\bigoplus_{(p,r,m)\in\mathcal{G}(G)} a_{p,r,m}+_{p^r} b_{p,r,m}
\end{align*}
where $+$ denotes the group operation and $+_{p^r}$ denotes addition mod-$p^r$. More generally, let $a,b,\cdots,z$ be any number of elements of $G$. Then we have
\begin{align}\label{eqn:a_to_z_summation}
a+b+\cdots+z=\bigoplus_{(p,r,m)\in\mathcal{G}(G)} \left(a_{p,r,m}+_{p^r} b_{p,r,m}+_{p^r} \cdots+_{p^r} z_{p,r,m}\right)
\end{align}
This can equivalently be written as
\begin{align*}
\left[a+b+\cdots+z\right]_{p,r,m}= a_{p,r,m}+_{p^r} b_{p,r,m}+_{p^r} \cdots+_{p^r} z_{p,r,m}
\end{align*}
where $[\cdot]_{p,r,m}$ denotes the $(p,r,m)\textsuperscript{th}$ component of it's argument.

Let $\mathbb{I}_{G:p,r,m}\in G$ be a generator for the group which is isomorphic to the $(p,r,m)\textsuperscript{th}$ ring of $G$. Then we have
\begin{align}\label{eqn:generator_representaion}
a=\overbrace{\sum}_{(p,r,m)\in\mathcal{G}(G)}^{(G)} a_{p,r,m} \mathbb{I}_{G:p,r,m}
\end{align}
where the summations are done with respect to the group operation and the multiplication $a_{p,r,m} \mathbb{I}_{G:p,r,m}$ is by definition the summation (with respect to the group operation) of $\mathbb{I}_{G:p,r,m}$ to itself for $a_{p,r,m}$ times. In other words, $a_{p,r,m} \mathbb{I}_{G:p,r,m}$ is the short hand notation for
\begin{align*}
a_{p,r,m} \mathbb{I}_{G:p,r,m}=\overbrace{\sum}_{i\in\{1,\cdots,a_{p,r,m}\}}^{(G)} \mathbb{I}_{G:p,r,m}
\end{align*}
where the summation is the group operation.

\noindent \textbf{Example:} Let $G=\mathds{Z}_4 \oplus
\mathds{Z}_3\oplus\mathds{Z}_9^2$. Then we have
$\mathcal{P}(G)=\{2,3\}$, $S_2(G)=\mathds{Z}_4$ and
$S_3(G)=\mathds{Z}_3\oplus\mathds{Z}_9^2$, $\mathcal{R}_2(G)=\{2\}$,
$\mathcal{R}_3(G)=\{1,2\}$, $M_{2,2}=1$, $M_{3,1}=1$, $M_{3,2}=2$ and
\begin{align*}
\mathcal{G}(G)=\{(2,2,1),(3,1,1),(3,2,1),(3,2,2)\}
\end{align*}
Each element $a$ of $G$ can be represented by a quadruple $(a_{2,2,1},a_{3,1,1},a_{3,2,1},a_{3,2,2})$ where $a_{2,2,1}\in\mathds{Z}_4$, $a_{3,1,1}\in\mathds{Z}_3$ and $a_{3,2,1},a_{3,2,2}\in\mathds{Z}_9$. Finally, we have $\mathbb{I}_{G:2,2,1}=(1,0,0,0)$, $\mathbb{I}_{G:3,1,1}=(0,1,0,0)$, $\mathbb{I}_{G:3,2,1}=(0,0,1,0)$, $\mathbb{I}_{G:3,2,2}=(0,0,0,1)$ so that Equation \eqref{eqn:generator_representaion} holds.

In the following section, we introduce the ensemble of Abelian group codes which we use in the paper.
\subsection{The Image Ensemble}\label{Image_Ensemble}
Recall that for a positive integer $n$, an Abelian group code of length $n$ over the group $G$ is a subgroup of $G^n$. Our ensemble of codes consists of all Abelian group codes over $G$; i.e. we consider all subgroups of $G^n$. We use the following fact to characterize all subgroups of $G^n$:

\begin{lemma}
For an Abelian group $\tilde{G}$, let $\phi:J\rightarrow \tilde{G}$ be a homomorphism from some Abelian group $J$ to $\tilde{G}$. Then $\phi(J)\le \tilde{G}$; i.e. the image of the homomorphism is a subgroup of $\tilde{G}$. Moreover, for any subgroup $\tilde{H}$ of $\tilde{G}$ there exists a corresponding Abelian group $J$ and a homomorphism $\phi:J\rightarrow \tilde{G}$ such that $\tilde{H}=\phi(J)$.
\end{lemma}
\begin{proof}
The first part of the lemma is proved in \cite[Theorem 12-1]{algebra_bloch}. For the second part, Let $J$ be isomorphic to $\tilde{H}$ and let $\phi$ be the identity mapping (more rigorously, let $\phi$ be the isomorphism between $J$ and $\tilde{H}$).
\end{proof}

In order to use the above lemma to construct the ensemble of subgroups of $G^n$, we need to identify all groups $J$ from which there exist non-trivial homomorphisms to $G^n$. Then the above lemma implies that for each such $J$ and for each homomorphism $\phi:J\rightarrow G^n$, the image of the homomorphism is a group code over $G$ of length $n$ and for each group code $\mathds{C}\le G^n$, there exists a group $J$ and a homomorphism such that $\mathds{C}$ is the image of the homomorphism. This ensemble corresponds to the ensemble of linear codes characterized by their generator matrix when the underlying group is a field of prime size. Note that as in the case of standard ensembles of linear codes, the correspondence between this ensemble and the set of Abelian group codes over $G$ of length $n$ may not be one-to-one.\\

Let $\tilde{G}$ and $J$ be two Abelian groups with decompositions:
\begin{align*}
&\tilde{G}=\bigoplus_{(p,r,m)\in\mathcal{G}(\tilde{G})} \mathds{Z}_{p^r}^{(m)}\\
&J=\bigoplus_{(q,s,l)\in\mathcal{G}(J)} \mathds{Z}_{q^s}^{(l)}
\end{align*}
and let $\phi$ be a homomorphism from $J$ to $\tilde{G}$. For $(q,s,l)\in\mathcal{G}(J)$ and $(p,r,m)\in\mathcal{G}(\tilde{G})$, let
\begin{align*}
g_{(q,s,l)\rightarrow (p,r,m)}=[\phi(\mathbb{I}_{J:q,s,l})]_{p,r,m}
\end{align*}
where $\mathbb{I}_{J:q,s,l}\in J$ is the standard generator for the $(q,s,l)\textsuperscript{th}$ ring of $J$ and $[\phi(\mathbb{I}_{J:q,s,l})]_{p,r,m}$ is the $(p,r,m)\textsuperscript{th}$ component of $\phi(\mathbb{I}_{J:q,s,l})\in\tilde{G}$. For $a=\bigoplus_{(q,s,l)\in\mathcal{G}(J)} a_{q,s,l}\in J$, let $b=\phi(a)$ and write $b=\bigoplus_{(p,r,m)\in\mathcal{G}(\tilde{G})} b_{p,r,m}$. Note that as in Equation \eqref{eqn:generator_representaion}, we can write:
\begin{align*}
a&=\overbrace{\sum}_{(q,s,l)\in\mathcal{G}(J)}^{(J)} a_{q,s,l} \mathbb{I}_{J:q,s,l}\\
&=\overbrace{\sum}_{(q,s,l)\in\mathcal{G}(J)}^{(J)} \overbrace{\sum}_{i\in\{1,\cdots,a_{q,s,l}\}}^{(J)} \mathbb{I}_{J:q,s,l}
\end{align*}
where the summations are the group summations. We have
\begin{align*}
b_{p,r,m}&=\left[\phi(a)\right]_{p,r,m}\\
&=\left[ \phi\left(\overbrace{\sum}_{(q,s,l)\in\mathcal{G}(J)}^{(J)} \overbrace{\sum}_{i\in\{1,\cdots,a_{q,s,l}\}}^{(J)} \mathbb{I}_{J:q,s,l}\right)\right]_{p,r,m}\\
&\stackrel{(a)}{=}\left[\overbrace{\sum}_{(q,s,l)\in\mathcal{G}(J)}^{(\tilde{G})} \overbrace{\sum}_{i\in\{1,\cdots,a_{q,s,l}\}}^{(\tilde{G})}  \phi\left(\mathbb{I}_{J:q,s,l}\right)\right]_{p,r,m}\\
&\stackrel{(b)}{=}\overbrace{\sum}_{(q,s,l)\in\mathcal{G}(J)}^{(\mathds{Z}_{p^r})}  \overbrace{\sum}_{i\in\{1,\cdots,a_{q,s,l}\}}^{(\mathds{Z}_{p^r})}  \left[\phi\left(\mathbb{I}_{J:q,s,l}\right)\right]_{p,r,m}\\
&\stackrel{(c)}{=}\overbrace{\sum}_{(q,s,l)\in\mathcal{G}(J)}^{(\mathds{Z}_{p^r})} a_{q,s,l} \left[\phi\left(\mathbb{I}_{J:q,s,l}\right)\right]_{p,r,m}\\
&=\overbrace{\sum}_{(q,s,l)\in\mathcal{G}(J)}^{(\mathds{Z}_{p^r})} a_{q,s,l} g_{(q,s,l)\rightarrow (p,r,m)}
\end{align*}
Note that $(a)$ follows since $\phi$ is a homomorphism; $(b)$ follows from Equation \eqref{eqn:a_to_z_summation}; and $(c)$ follows by using $a_{q,s,l} \left[\phi\left(\mathbb{I}_{J:q,s,l}\right)\right]_{p,r,m}$ as the short hand notation for the summation of $\left[\phi\left(\mathbb{I}_{J:q,s,l}\right)\right]_{p,r,m}$ to itself for $a_{q,s,l}$ times. %$\overbrace{\sum}_{i\in\{1,\cdots,a_{q,s,l}\}}^{(\mathds{Z}_{p^r})} \left[\phi\left(\mathbb{I}_{J:q,s,l}\right)\right]_{p,r,m}$.\\

Note that $g_{(q,s,l)\rightarrow (p,r,m)}$ represents the effect of the $(q,s,l)\textsuperscript{th}$ component of $a$ on the $(p,r,m)\textsuperscript{th}$ component of $b$ dictated by the homomorphism. This means that the homomorphism $\phi$ can be represented by
\begin{align}\label{eqn:hom_general_form1}
\phi(a)=\bigoplus_{(p,r,m)\in\mathcal{G}(\tilde{G})} \overbrace{\sum}_{(q,s,l)\in\mathcal{G}(J)}^{(\mathds{Z}_{p^r})} a_{q,s,l} g_{(q,s,l)\rightarrow (p,r,m)}
\end{align}
where $a_{q,s,l} g_{(q,s,l)\rightarrow (p,r,m)}$ is the short-hand notation for the mod-$p^r$ addition of $g_{(q,s,l)\rightarrow (p,r,m)}$ to itself for $a_{q,s,l}$ times. We have the following lemma on $g_{(q,s,l)\rightarrow (p,r,m)}$:

\begin{lemma}\label{lemma:g}
For a homomorphism described by \eqref{eqn:hom_general_form1}, we have
\begin{align*}
\begin{array}{ll}
g_{(q,s,l)\rightarrow (p,r,m)}=0&\mbox{If }p\ne q\\
%g_{(q,s,l)\rightarrow (p,r,m)}\in \mathds{Z}_{p^r}&\mbox{If }p= q, r\le s\\
g_{(q,s,l)\rightarrow (p,r,m)}\in p^{r-s}\mathds{Z}_{p^r}&\mbox{If }p= q, r\ge s
\end{array}
\end{align*}
Moreover, any mapping described by \eqref{eqn:hom_general_form1} and satisfying these conditions is a homomorphism.
\end{lemma}
\begin{proof}
The proof is provided in Appendix \ref{section:Proof_Image_Ensemble}.
\end{proof}

This lemma implies that in order to construct a subgroup of $\tilde{G}$, we only need to consider homomorphisms from an Abelian group $J$ to $\tilde{G}$ such that %that for two Abelian groups $\tilde{G}$ and $J$ if $\mathcal{P}(\tilde{G})$ and $\mathcal{P}(J)$ are disjoint, then there does not exist non-trivial homomorphisms between the two groups. It further implies
\begin{align*}
\mathcal{P}(J)\subseteq \mathcal{P}(\tilde{G})
\end{align*}
since if for some $(q,s,l)\in\mathcal{G}(J)$, $q\notin \mathcal{P}(\tilde{G})$ then $\phi(a)$ would not depend on $a_{q,s,l}$. For $p\in\mathcal{P}(\tilde{G})$, define
\begin{align}\label{eqn:r_p}
r_{p}=\max \mathcal{R}_p(G)%\max_{(p,r',m')\in\mathcal{G}(\tilde{G})} r'
\end{align}
We show that we can restrict ourselves to $J$'s such that for all $(q,s,l)\in\mathcal{G}(J)$, $s\le r_q$. % For $(q,s,l)\in\mathcal{G}(J)$, assume $s>r_p$. Then for all $(p,r,m)\in\mathcal{G}(\tilde{G})$, if $p=q$, we have $s>r$.
%Recall that
%\begin{align*}
%\phi(a)=\bigoplus_{(p,r,m)\in\mathcal{G}(\tilde{G})} \overbrace{\sum}_{(q,s,l)\in\mathcal{G}(J)}^{\mathds{Z}_{p^r}} a_{q,s,l} g_{(q,s,l)\rightarrow (p,r,m)}
%\end{align*}
Let $(p,r,m)\in\mathcal{G}(\tilde{G})$ be such that $p=q$. Since $g_{(q,s,l)\rightarrow (p,r,m)}\in\mathds{Z}_{p^r}$ and $r\le r_q$, we have
\begin{align*}
\left(a_{q,s,l} g_{(q,s,l)\rightarrow (p,r,m)}\right) \pmod{p^r}&=\left(\left(a_{q,s,l} \right) \pmod{p^r} g_{(q,s,l)\rightarrow (p,r,m)}\right) \pmod{p^r}\\
&=\left(\left(a_{q,s,l} \right) \pmod{p^{r_q}} g_{(q,s,l)\rightarrow (p,r,m)}\right) \pmod{p^r}
\end{align*}
This implies that for all $a\in J$ and all $(q,s,l)\in\mathcal{G}(J)$, in the expression for the $(p,r,m)\textsuperscript{th}$ component of $\phi(a)$ with $p=q$, $a_{q,s,l}$ appears as $\left(a_{q,s,l} \right) \pmod{q^{r_q}}$. Therefore, it suffices for $a_{q,s,l}$ to take values from $\mathds{Z}_{q^{r_q}}$ and this happens if $s\le r_q$.\\
%and this is valid for all ${q,s,l}\in\mathcal{G}(J)$. Therefore we only need to consider $J$'s such that for $a\in J$ and for $(q,s,l)\in\mathcal{G}(J)$, $a_{q,s,l}\in \mathds{Z}_{p^r}$ for some $(p,r,m)\in\mathcal{G}(\tilde{G})$. This implies that we only need to consider the case where $s\le r_p$.\\

To construct Abelian group codes of length $n$ over $G$, let $\tilde{G}=G^n$. we have
\begin{align}\label{eqn:Gn}
G^n&\cong\bigoplus_{p\in\mathcal{P}(G)} \bigoplus_{r\in\mathcal{R}_p} \mathds{Z}_{p^r}^{n M_{p,r}} =\bigoplus_{p\in\mathcal{P}(G)} \bigoplus_{r\in\mathcal{R}_p} \bigoplus_{m=1}^{n M_{p,r}} \mathds{Z}_{p^r}^{(m)}
=\bigoplus_{(p,r,m)\in\mathcal{G}(G^n)} \mathds{Z}_{p^r}^{(m)}
\end{align}

Define $J$ as
\begin{align}\label{eqn:J}
J=\bigoplus_{q\in\mathcal{P}(G)} \bigoplus_{s=1}^{r_q} \mathds{Z}_{q^s}^{k_{q,s}} =\bigoplus_{q\in\mathcal{P}(G)} \bigoplus_{s=1}^{r_q} \bigoplus_{l=1}^{k_{q,s}} \mathds{Z}_{q^s}^{(l)} =\bigoplus_{(q,s,l)\in\mathcal{G}(J)} \mathds{Z}_{q^s}^{(l)}
\end{align}
for some positive integers $k_{q,s}$.

\noindent \textbf{Example:} Let $G=\mathds{Z}_8 \oplus \mathds{Z}_9
\oplus \mathds{Z}_5$. Then we have
\[
J=\mathds{Z}_2^{k_{2,1}} \oplus \mathds{Z}_4^{k_{2,2}} \oplus
\mathds{Z}_8^{k_{2,3}} \oplus \mathds{Z}_3^{k_{3,1}} \oplus
\mathds{Z}_9^{k_{3,2}} \oplus \mathds{Z}_5^{k_{5,1}}
\]

Define
\begin{align*}
k=\sum_{q\in \mathcal{P}(G)} \sum_{s=1}^{r_q} k_{q,s}
\end{align*}
and $w_{q,s}=\frac{k_{q,s}}{k}$ for $q\in\mathcal{P}(G)$ and $s=1,\cdots,r_q$ so that we can write
\begin{align}\label{eqn:J2}
J=\bigoplus_{q\in\mathcal{P}(G)} \bigoplus_{s=1}^{r_q} \bigoplus_{l=1}^{k w_{q,s}} \mathds{Z}_{q^s}^{(l)}
\end{align}
for some constants $w_{q,s}$ adding up to one.

The ensemble of Abelian group encoders consists of all mappings $\phi:J\rightarrow G^n$ of the form% Let $k=\sum_{q\in\mathcal{P}(G)} \sum_{s=1}^{r_q} k_{q,s}$ and define $w_{q,s}=\frac{k_{q,s}}{k}$.
\begin{align}\label{eqn:phi}
\phi(a)=\bigoplus_{(p,r,m)\in\mathcal{G}(G^n)} \overbrace{\sum}_{(q,s,l)\in\mathcal{G}(J)}^{(\mathds{Z}_{p^r})} a_{q,s,l} g_{(q,s,l)\rightarrow (p,r,m)}
\end{align}
for $a\in J$ where $g_{(q,s,l)\rightarrow (p,r,m)}=0$ if $p\ne q$, $g_{(q,s,l)\rightarrow (p,r,m)}$ is a uniform random variable over $\mathds{Z}_{p^r}$ if $p=q, r\le s$, and $g_{(q,s,l)\rightarrow (p,r,m)}$ is a uniform random variable over $p^{r-s}\mathds{Z}_{p^r}$ if $p=q, r\ge s$. The corresponding shifted group code is defined by
\begin{align}\label{eqn:Code}
\mathds{C}=\{\phi(a)+B|a\in J\}
\end{align}
where $B$ is a uniform random variable over $G^n$. The rate of this code is given by
\begin{align}\label{eqn:rate}
R&=\frac{1}{n}\log |J|=\frac{k}{n}\sum_{q\in\mathcal{P}(G)} \sum_{s=1}^{r_q} s w_{q,s} \log q
\end{align}

\begin{remark}
An alternate approach to constructing Abelian group codes is to consider kernels of homomorphisms (the kernel ensemble). To construct the ensemble of Abelian group codes in this manner, let $\phi$ be a homomorphism from $J$ into $G^n$ such that for $a\in G^n$,
\begin{align*}
\phi(a)=\bigoplus_{(q,s,l)\in\mathcal{G}(J)} \overbrace{\sum}_{(p,r,m)\in\mathcal{G}(G^n)}^{(\mathds{Z}_{q^s})} a_{p,r,m} g_{(p,r,m)\rightarrow (q,s,l)}
\end{align*}
where $g_{(p,r,m)\rightarrow (q,s,l)}=0$ if $q\ne p$, $g_{(p,r,m)\rightarrow (q,s,l)}$ is a uniform random variable over $\mathds{Z}_{q^s}$ if $q=p, s\le r$, and $g_{(p,r,m)\rightarrow (q,s,l)}$ is a uniform random variable over $p^{s-r}\mathds{Z}_{q^s}$ if $q=p, s\ge r$. The code is given by $\mathds{C}=\{a\in G^n|\phi(a)=c\}$ where $c$ is a uniform random variable over $J$.\\
In this paper, we use the image ensemble for both the channel and the source coding problem; however, similar results can be derived using the kernel ensemble as well.
\end{remark}

%%%%%%%%%%%%%%%%%%%%%%%%%%%%%%%%%%%%%%%%%%%%%%%%%%%%%%%%%%%%%%%%%%%%%%%%%%%%%%%%%%%%%%%%%%%%%%%%%%%%%%%%%

\section{Main Results}\label{section:main_results}
In this section, we provide  an upper bound on the optimal
rate-distortion function for a given source  and a lower bound on the
capacity of a given  channel using group codes when the underlying group
is an arbitrary Abelian group represented by Equation
\eqref{eqn:G}.
We start by defining seven objects and then state two theorems
using these objects, and  finally provide an interpretation of the results and these objects
with two examples.

\subsection{Definitions}

For $q\in\mathcal{P}(G)$, let
$\mathcal{S}_q(G)=\{1,2,\cdots,r_q\}$ where $r_q$ is defined as
\begin{align}\label{eqn:r_p2}
r_q=\max \mathcal{R}_q(G)
\end{align}
Define
\begin{align}\label{eqn:S_G}
\mathcal{S}(G)=\{(q,s)|q\in\mathcal{P}(G), s\in\mathcal{S}_{q}(G)\}%\bigotimes_{q\in\mathcal{P}(G)} \mathcal{S}_q(G)
\end{align}
\begin{align}\label{eqn:Q_G}
\mathcal{Q}(G)=\{(p,r)|p\in\mathcal{P}(G),r\in\mathcal{R}_p(G)\}
\end{align}
We denote vectors $\hat{\theta}$ and $w$ whose components are indexed by $(q,s)\in\mathcal{S}(G)$ by $(\hat{\theta}_{q,s})_{(q,s)\in\mathcal{S}(G)}$
and $(w_{q,s})_{(q,s)\in\mathcal{S}(G)}$ respectively and a vector $\theta$ whose components are indexed by $(p,r) \in\mathcal{Q}(G)$
by $(\theta_{p,r})_{(p,r)\in\mathcal{Q}(G)}$. For $\hat{\theta}=(\hat{\theta}_{q,s})_{(q,s)\in\mathcal{S}(G)}$, define
\begin{align*}
\pmb{\theta}(\hat{\theta})=\left(\min_{\substack{(q,s)\in\mathcal{S}(G)\\q=p\\w_{q,s}\ne 0}} |r-s|^+ +\hat{\theta}_{q,s}\right)_{(p,r)\in \mathcal{Q}(G)}
\end{align*}
and let
\begin{align}\label{eqn:Theta_S}
\Theta(w)=\left\{\pmb{\theta}(\hat{\theta}) \left|(\hat{\theta}_{q,s})_{(q,s)\in\mathcal{S}(G)}:0 \le \hat{\theta}_{q,s}\le s\right.\right\}
\end{align}
This set corresponds to a collection of subgroups of $G$ which appear in the rate-distortion function. In other words, depending on the weights $w$, certain
subgroups of the group become important in the rate-distortion function. This will be clarified in the proof of the theorem. For $\theta\in\Theta(w)$, define
 \begin{align}\label{eqn:omega_theta}
\omega_{\theta} &=\frac{\sum_{(q,s)\in\mathcal{S}} \displaystyle{\max_{\substack{(p,r)\in\mathcal{Q}(G)\\p=q}}} \left(\theta_{p,r}-|r-s|^+\right)^+ w_{q,s} \log q} {\sum_{(q,s)\in\mathcal{S}} s w_{q,s} \log q}
\end{align}
and let $H_{\theta}$ be a subgroup of $G$ defined as
\begin{align}\label{eqn:H_theta}
H_{\theta}=\bigoplus_{(p,r,m)\in\mathcal{G}(G)} p^{\theta_{p,r}} \mathds{Z}_{p^r}^{(m)}
\end{align}
Let $X$ and $U$ be jointly distributed random variables such that $U$ is uniform over $G$ and let $[U]_{\theta}=U+H_{\theta}$ be a random variable taking values from the cosets of $H_{\theta}$ in $G$. We define \emph{the source coding group mutual information} between $U$ and $X$ as
\begin{align}\label{eqn:Rate_Source}
I_{s.c.}^{G}(U;X)=\min_{\substack{w_{q,s}, (q,s)\in\mathcal{S}(G)\\\sum w_{q,s}=1}} \max_{\substack{\theta\in\Theta(w)\\\theta\ne \pmb{0}}} \frac{1}{\omega_{\theta}} I([U]_{\theta};X)
\end{align}
where $\pmb{0}$ is a vector whose components are indexed by $(p,r)\in\mathcal{Q}(G)$ and whose $(p,r)\textsuperscript{th}$ component is equal to $0$.

Let $X$ and $Y$ be jointly distributed random variables such that $X$
is uniform over $G$ and let $[X]_{\theta}=X+H_{\theta}$ be a random
variable taking values from the cosets of $H_{\theta}$ in $G$. We
define \emph{the channel coding group mutual information} between $X$
and $Y$ as
\begin{align}\label{eqn:Rate_Channel}
I_{c.c.}^{G}(X;Y)=\max_{\substack{w_{q,s}, (q,s)\in\mathcal{S}(G)\\\sum w_{q,s}=1}} \min_{\substack{\theta\in\Theta(w)\\\theta\ne \pmb{r}}} \frac{1}{1-\omega_{\theta}} I(X;Y|[X]_{\theta})
\end{align}
where $\pmb{r}$ is a vector whose components are indexed by $(p,r)\in\mathcal{Q}(G)$ and whose $(p,r)\textsuperscript{th}$ component is equal to $r$.

\subsection{Main Results}

The following theorem is the first main result of this paper.
\begin{theorem}
For a source $(\mathcal{X},\mathcal{U}=G,p_X,d)$ and a given
distortion level $D$, let $p_{XU}$ be a joint distribution over
$\mathcal{X}\times \mathcal{U}$ such that its first marginal is equal
to the source distribution $p_X$, its second marginal $p_U$ is uniform
over $\mathcal{U}=G$ and such that $\mathds{E}\{d(X,U)\}\le D$.
Then the rate-distortion pair $(R,D)$ is achievable where $R=I_{s.c.}^{G}(U;X)$.
\end{theorem}
\begin{proof}
The proof is provided in Section \ref{section:ErrAnalysis_General_Source}.
\end{proof}

When the underlying group is a $\mathds{Z}_{p^r}$ ring, this result can be simplified. We state this result in the form of a corollary:

\begin{cor}\label{cor1}
Let $X$, $U$ be jointly distributed random variables such that $U$ is uniform over $\mathcal{U}=G=\mathds{Z}_{p^r}$ for some prime $p$ and positive integer $r$. For $\theta=1,2,\cdots,r$, let $H_{\theta}$ be a subgroup of $\mathds{Z}_{p^r}$ defined by $H_{\theta}=p^{\theta}\mathds{Z}_{p^r}$ and let $[U]_{\theta}=U+H_{\theta}$. Then,
\begin{align*}
I_{s.c.}^{G}(U;X)=\max_{\theta=1}^r \frac{r}{\theta} I([U]_{\theta};X)
\end{align*}

%For a source $(\mathcal{X},\mathcal{U}=\mathds{Z}_{p^r},p_X,d)$ and a
%given distortion level $D$, let $p_{XU}$ be a joint distribution over
%$\mathcal{X}\times \mathcal{U}$ such that its first marginal is equal
%to the source distribution $p_X$, its second marginal $p_U$ is
%uniform over $\mathcal{U}=\mathds{Z}_{p^r}$ and such that
%$\mathds{E}\{d(X,U)\}\le D$. For $\theta=1,2,\cdots,r$, let
%$H_{\theta}$ be a subgroup of $G$ defined by
%$H_{\theta}=p^{\theta}\mathds{Z}_{p^r}$ and let
%$[U]_{\theta}=U+H_{\theta}$. The rate-distortion pair $(R,D)$ is
%achievable where
%\begin{align}\label{eqn:Rate_Source}
%R= \max_{\theta=1}^r \frac{1}{\theta} I([U]_{\theta};X)
%\end{align}
\end{cor}
\begin{proof}
The proof is provided in Section \ref{section:simplification}.
\end{proof}

The following theorem is the second main result of this paper.
\begin{theorem}
For a channel $(\mathcal{X}=G,\mathcal{Y},W_{Y|X})$, the rate $R=I_{c.c.}^{G}(X;Y)$ is achievable using group codes over $G$.
\end{theorem}
\begin{proof}
The proof is provided in Section \ref{section:ErrAnalysis_General_Channel}.
\end{proof}

When the underlying group is a $\mathds{Z}_{p^r}$ ring, this result can be simplified. We state this result in the form of a corollary:

\begin{cor}\label{cor2}
Let $X$, $Y$ be jointly distributed random variables such that $X$ is uniform over $\mathcal{X}=G=\mathds{Z}_{p^r}$ for some prime $p$ and a positive integer $r$. For $\theta=0,1,\cdots,r-1$, let $H_{\theta}$ be a subgroup of $\mathds{Z}_{p^r}$ defined by $H_{\theta}=p^{\theta}\mathds{Z}_{p^r}$ and let $[X]_{\theta}=X+H_{\theta}$. Then,
\begin{align*}
I_{c.c.}^{G}(X;Y)=\max_{\theta=0}^{r-1} \frac{r}{r-\theta} I(X;Y|[X]_{\theta})
\end{align*}
\end{cor}
\begin{proof}
The proof is provided in Section \ref{section:simplification_channel}.
\end{proof}

When dealing with group codes for the purpose of channel coding, an important case is when the channel exhibits some sort of symmetry. The capacity of group codes for channels with some notion of symmetry is found in \cite{fagnani_abelian}. The next corollary states that the result of this paper simplifies to the result of \cite{fagnani_abelian} when the channel is symmetric in the sense defined in \cite{fagnani_abelian}.

\begin{cor}\label{cor:symmetric}
When the channel $(\mathcal{X}=G,\mathcal{Y},W_{Y|X})$ is $G$-symmetric in the sense defined in \cite{fagnani_abelian}, i.e. if
\begin{enumerate}
\item $G$ acts simply transitively on $\mathcal{X}$ (trivially holds for this case)
\item $G$ acts isometrically on $\mathcal{Y}$
\item For all $x,g\in G, y\in\mathcal{Y}$, $W(y|x)=W(g\cdot y|g+x)$
\end{enumerate}
then $I_{c.c}^G(X;Y)$ is equal to the rate provided in \cite[Equation (33)]{fagnani_abelian}.
\end{cor}
\begin{proof}
The proof is provided in Section \ref{section:simplification_symmetric}.
\end{proof}

\subsection{Interpretation of the Result}

In this section, we try to give some intuition about the result and
the quantities defined above using several examples.
At a high level, $w_{q,s}$ denotes the normalized weight given to the
$\mathds{Z}_{q^s}$ component
of the input group $J$ in constructing the homomorphism from $J$ to
$G^n$, and $\theta$ indexes a subgroup $H_{\theta}$ of $G$ that comes from a
collection $\Theta(w)$ governed by the choice of
$w_{q,s}s$. $\frac{1}{\omega{\theta}} I([U]_{\theta};X)$ in source
coding and $\frac{1}{(1-\omega{\theta})} I(X;Y|[X]_{\theta})$ in
channel coding  denote the
rate constraints imposed by the subgroup $H_{\theta}$. Due to the
algebraic structure of the code, in the ensemble  two random
codewords corresponding to two distinct indexes
are statistically dependent, unless $G$ is a finite field.
For the source coding problem, when the code is chosen randomly,
consider the event that all components
of their difference belong  to a proper subgroup
$H_{\theta}$ of $G$.  Then if one of them is a poor representation of
a given source sequence, so is the other with a probability that is
higher than the case when no algebraic structure on the code is enforced.
This means that the code size has to be larger so that with high probability one can find
a good representation of the source.
For the channel coding problem, when a random codeword corresponding
to a given message index is transmitted over the channel, consider the event
that all components of the difference between the codeword transmitted
and a random codeword corresponding to another message index belong to
a proper subgroup $H_{\theta}$ of $G$. Then  the probability
that the latter is decoded instead of the former  is higher than the
case when no algebraic structure on the code is enforced.

\noindent \textbf{Example:} We start with the
simple example where $G=\mathds{Z}_8$. In this case, we have $\mathcal{P}(G)=\{2\}$, $r_2=3$, $\mathcal{S}_2=\{1,2,3\}$, $\mathcal{S}=\{(2,1),(2,2),(2,3)\}$,
and $\mathcal{Q}(G)=\{(2,3)\}$. For vectors $w$, $\hat{\theta}$ and $\theta$ defined as above, we have $w=(w_{2,1},w_{2,2},w_{2,3})$,
$\hat{\theta}=(\hat{\theta}_{2,1},\hat{\theta}_{2,2},\hat{\theta}_{2,3})$ and $\theta=\theta_{2,3}$. Recall that the ensemble of Abelian group codes used
in the random coding argument consists of the set of all homomorphisms from some
$J=\mathds{Z}_{2}^{kw_{2,1}}\oplus  \mathds{Z}_{4}^{kw_{2,2}}
\oplus \mathds{Z}_{8}^{kw_{2,3}}$ and hence the vector of weights $w$ determines the input
group of the homomorphism. Depending on the values of the weights, the structure of the input group can be different; for example, if $w_{2,1}=0$, $w_{2,2}\ne 0$
and $w_{2,3}\ne 0$, the input group will only have $\mathds{Z}_4$ and $\mathds{Z}_8$ components. Any vector
$\hat{\theta}=(\hat{\theta}_{2,1},\hat{\theta}_{2,2},\hat{\theta}_{2,3})$ with $0\le \hat{\theta}_{2,1}\le 1$, $0\le \hat{\theta}_{2,2}\le 2$
and $0\le \hat{\theta}_{2,3}\le 3$ corresponds to a subgroup $K_{\hat{\theta}}$ of the input group $J$ given by
\begin{align*}
K_{\hat{\theta}}=2^{\hat{\theta}_{2,1}}\mathds{Z}_{2}^{kw_{2,1}}\oplus 2^{\hat{\theta}_{2,2}} \mathds{Z}_{4}^{kw_{2,2}}
\oplus 2^{\hat{\theta}_{2,3}}\mathds{Z}_{8}^{kw_{2,3}}
\end{align*}
Similarly, any
$\theta=\theta_{2,3}$ with $0\le \theta_{2,3}\le 3$ corresponds to a subgroup $H_{\theta}$ of the group space $G^n$ given by
\begin{align*}
H_{{\theta}}=2^{{\theta}_{2,3}}\mathds{Z}_{8}^n
\end{align*}
Let us assume $w_{2,1}=0$, $w_{2,2}\ne 0$ and $w_{2,3}\ne 0$ so that $K_{\hat{\theta}}=2^{\hat{\theta}_{2,2}} \mathds{Z}_{4}^{kw_{2,2}}
\oplus 2^{\hat{\theta}_{2,3}}\mathds{Z}_{8}^{kw_{2,3}}$. It turns out that if
\begin{align}
\theta=\pmb{\theta}(\hat{\theta})=\min\left(1+\hat{\theta}_{2,2},\hat{\theta}_{2,3}\right)
\end{align}
then for any random homomorphism $\phi$ from $J$ into $G^n$, and for any $a=(\alpha,\beta)\in J$ with
$\alpha\in 2^{\hat{\theta}_{2,2}} \mathds{Z}_{4}^{kw_{2,2}}\backslash  2^{\hat{\theta}_{2,2}+1} \mathds{Z}_{4}^{kw_{2,2}}$ and
$\beta\in 2^{\hat{\theta}_{2,3}} \mathds{Z}_{8}^{kw_{2,3}}\backslash  2^{\hat{\theta}_{2,3}+1} \mathds{Z}_{8}^{kw_{2,3}}$, $\phi(a)$
is uniformly distributed over $H_{{\theta}}^n$.
The set $\Theta(w)$ consists of all vectors $\theta$ for which there exists at least one such $a$. Note that this set corresponds to a collection of subgroups of $G^n$. The quantity $1-\omega_{\theta}$
is a measure of the number of elements $a$ of $J$ for which $\phi(a)$ is uniform over $H_{\theta}$. It turns out that for this example,
$\Theta(w)=\{0,1,2,3\}$ and $\omega_0=0$, $\omega_1=\frac{w_{2,3}}{2w_{2,2}+3w_{2,3}}$, $\omega_2=\frac{w_{2,2}+2w_{2,3}}{2w_{2,2}+3w_{2,3}}$
and $\omega_3=1$.

\noindent \textbf{Example:} Next, we consider the case where
$G=\mathds{Z}_4\oplus \mathds{Z}_3$. In this case, we have
$\mathcal{P}(G)=\{2,3\}$, $r_2=2$, $r_3=1$,
$\mathcal{S}_2=\{1,2\}$, $\mathcal{S}_3=\{1\}$, $\mathcal{S}=\{(2,1),(2,2),(3,1)\}$,
and $\mathcal{Q}(G)=\{(2,2),(3,1)\}$. For vectors $w$, $\hat{\theta}$ and $\theta$ defined as before, we have $w=(w_{2,1},w_{2,2},w_{3,1})$,
$\hat{\theta}=(\hat{\theta}_{2,1},\hat{\theta}_{2,2},\hat{\theta}_{3,1})$ and $\theta=(\theta_{2,2},\theta_{3,1})$.
The ensemble of Abelian group codes consists of the set of all homomorphisms from some
$J=\mathds{Z}_{2}^{kw_{2,1}}\oplus  \mathds{Z}_{4}^{kw_{2,2}}
\oplus \mathds{Z}_{3}^{kw_{3,1}}$. Any vector
$\hat{\theta}=(\hat{\theta}_{2,1},\hat{\theta}_{2,2},\hat{\theta}_{3,1})$ with $0\le \hat{\theta}_{2,1}\le 1$, $0\le \hat{\theta}_{2,2}\le 2$
and $0\le \hat{\theta}_{3,1}\le 1$ corresponds to a subgroup $K_{\hat{\theta}}$ of the input group $J$ given by
\begin{align*}
K_{\hat{\theta}}=2^{\hat{\theta}_{2,1}}\mathds{Z}_{2}^{kw_{2,1}}\oplus 2^{\hat{\theta}_{2,2}} \mathds{Z}_{4}^{kw_{2,2}}
\oplus 3^{\hat{\theta}_{3,1}}\mathds{Z}_{8}^{kw_{3,1}}
\end{align*}
Similarly, any
$\theta=(\theta_{2,2},\theta_{3,1})$ with $0\le \theta_{2,2}\le 2$ and $0\le \theta_{3,1}\le 1$ corresponds to a subgroup $H_{\theta}$
of the group space $G^n$ given by
\begin{align*}
H_{{\theta}}=2^{{\theta}_{2,2}}\mathds{Z}_{4}^n\oplus 3^{{\theta}_{3,1}}\mathds{Z}_{3}^n
\end{align*}
Let us assume $w_{2,1},w_{2,2},w_{2,3}$ are all non-zero. It turns out that if
\begin{align}
&\theta_{2,2}=\min\left(1+\hat{\theta}_{2,1},\hat{\theta}_{2,2}\right)\\
&\theta_{3,1}=\hat{\theta}_{3,1}
\end{align}
then for any random homomorphism $\phi$ from $J$ into $G^n$, and for any $a=(\alpha,\beta,\gamma)\in J$ with
$\alpha\in 2^{\hat{\theta}_{2,1}} \mathds{Z}_{2}^{kw_{2,1}}\backslash  2^{\hat{\theta}_{2,1}+1} \mathds{Z}_{2}^{kw_{2,1}}$,
$\beta\in 2^{\hat{\theta}_{2,2}} \mathds{Z}_{4}^{kw_{2,2}}\backslash  2^{\hat{\theta}_{2,2}+1} \mathds{Z}_{4}^{kw_{2,2}}$ and
$\gamma\in 3^{\hat{\theta}_{3,1}} \mathds{Z}_{3}^{kw_{3,1}}\backslash  3^{\hat{\theta}_{3,1}+1} \mathds{Z}_{3}^{kw_{3,1}}$, $\phi(a)$
is uniformly distributed over $H_{{\theta}}^n$.
Moreover, for this example we have
\begin{align*}
\Theta(w)=\{(0,0),(1,0),(2,0),(0,1),(1,1),(2,1)\}
\end{align*}

\section{Proof of Source Coding}\label{section:proof_source}

\subsection{Encoding and Decoding}
Following the analysis of Section \ref{Image_Ensemble}, we construct the ensemble of group codes of length $n$ over $G$ as the image of all homomorphisms $\phi$ from some Abelian group $J$ into $G^n$ where $J$ and $G^n$ are as in Equations \eqref{eqn:J2} and \eqref{eqn:Gn} respectively. The random homomorphism $\phi$ is described in Equation \eqref{eqn:phi}.

To find an achievable rate for a distortion level $D$, we use a random coding argument in which the random encoder is characterized by the random homomorphism $\phi$, a random vector $B$ uniformly distributed over $G^n$ and a joint distribution $p_{XU}$ over $\mathcal{X}\times \mathcal{U}$ such that its first marginal is equal to the source distribution $p_X$, its second marginal $p_U$ is uniform over $\mathcal{U}=G$ and such that $\mathds{E}\{d(X,U)\}\le D$. The code is defined as in \eqref{eqn:Code} and its rate is given by \eqref{eqn:rate}.

Given the source output sequence $x\in\mathcal{X}^n$, the random encoder looks for a codeword $u\in \mathds{C}$ such that $u$ is jointly typical with $x$ with respect to $p_{XU}$. If it finds at least one such $u$, it encodes $x$ to $u$ (if it finds more than one such $u$ it picks one of them at random). Otherwise, it declares error. The decoder outputs $u$ as the source reconstruction.
%by random matrices $G_s$ for $s=1,\cdots,r$ which are uniformly distributed over their domains and $B$ is a random element in $\mathds{Z}_p^n$.\\
%Given the codeword $u$, the decoder output .\\

\subsection{Error Analysis}\label{section:ErrAnalysis_General_Source}
Let $x=(x_1,\cdots,x_n)$ and $u=(u_1,\cdots,u_n)$ be the source output and the encoder/decoder output respectively. Note that if the encoder declares no error then since $x$ and $u$ are jointly typical, $(d(x_i,u_i))_{i=1,\cdots,n}$ is typical with respect to the distribution of $d(X,U)$. Therefore for large $n$, $\frac{1}{n}d(x,u)=\frac{1}{n}\sum_{i=1}^n d(x_i,u_i)\approx \mathds{E}\{d(X,U)\}\le D$. It remains to show that the rate can be as small as $I_{s.c.}^{G}(X;U)$ while keeping the probability of encoding error small.\\

Given the source output $x\in\mathcal{X}^n$, define
\begin{align*}
\alpha(x) =\sum_{u\in A_{\epsilon}^n(U|x)} \mathds{1}_{\{u\in\mathds{C}\}} =\sum_{u\in A_{\epsilon}^n(U|x)}\sum_{a\in J} \mathds{1}_{\{\phi(a)+B=u\}}
\end{align*}
An encoding error occurs if and only if $\alpha(x)=0$. We use the following Chebyshev's inequality to show that under certain conditions the probability of error can be made arbitrarily small:
\begin{align*}
P\left(\alpha(x)=0\right)\le \frac{\mbox{\small var}\{\alpha(x)\}}{\mathds{E}\{\alpha(x)\}^2}
\end{align*}
We need the following lemmas to proceed:

\begin{lemma}\label{lemma:Joint_Prob}
For $a,\tilde{a}\in J$, $u,\tilde{u}\in G^n$ and for $(q,s,l)\in\mathcal{G}(J)$, let $\hat{\theta}_{q,s,l}\in\{0,1,\cdots,s\}$ be such that
\begin{align*}
\tilde{a}_{q,s,l} - a_{q,s,l}\in q^{\hat{\theta}_{q,s,l}}\mathds{Z}_{q^s}\backslash q^{\hat{\theta}_{q,s,l}+1}\mathds{Z}_{q^s}
\end{align*}
%For $q\in\mathcal{P}(J)$ and $s=1,\cdots,r_q$, define
%\begin{align*}
%\hat{\theta}_{q,s}=\min{\substack{(q',s',l')\in\mathcal{G}(J)\\q'=q\\s'=s}} \hat{\theta}_{q',s',l'}
%\end{align*}
For $(p,r)\in\mathcal{Q}(G)$ define
\begin{align*}
\pmb{\theta}_{p,r}(a,\tilde{a})&=\min_{\substack{(q,s,l)\in\mathcal{G}(J)\\q=p}} \left|r-s\right|^+ +\hat{\theta}_{q,s,l}
%&=\min_{\substack{(q,s,l)\in\mathcal{G}(J)\\q=p}} \left|r-s\right|^+\hat{\theta}_{q,s}
\end{align*}
and let $\theta_{p,r}=\pmb{\theta}_{p,r}(a,\tilde{a})$. Define the subgroup $H_{\theta}$ of $G$ as
\begin{align*}
H_{\theta}=\bigoplus_{(p,r,m)\in\mathcal{G}(G)} p^{\theta_{p,r}} \mathds{Z}_{p^r}^{(m)}
\end{align*}
Then,
\begin{align*}
P\left(\phi(a)+B=u, \phi(\tilde{a})+B=\tilde{u}\right)= \left\{\begin{array}{ll}
\frac{1}{|G|^n}\frac{1}{|H_{\theta}|^n}& \mbox{If $\tilde{u}-u\in H_{\theta}^n$}\\
0&\mbox{Otherwise}
\end{array}\right.
\end{align*}
\end{lemma}
\begin{proof}
The proof is provided in Appendix \ref{section:Proof_Joint_Prob}
\end{proof}

\begin{lemma}\label{lemma:Joint_Prob_T}
For $a\in J$ and $\theta=(\theta_{p,r})_{(p,r)\in\mathcal{Q}(G)}$, let
\begin{align*}
T_{\theta}(a)=\{\tilde{a}\in J|\forall (p,r)\in\mathcal{Q}(G), \pmb{\theta}_{p,r}(a,\tilde{a})=\theta_{p,r}\}
\end{align*}
where $\pmb{\theta}_{p,r}(a,\tilde{a})$ is defined as in the previous lemma. Then we have
\begin{align*}
\left|T_{\theta}(a)\right|&\le \prod_{(q,s,l)\in\mathcal{G}(P)} q^{s- \displaystyle{\max_{\substack{(p,r)\in\mathcal{Q}(G)\\p=q}}}   \left(\theta_{p,r}-|r-s|^+\right)^+}\\
&=\prod_{(q,s)\in\mathcal{S}(G)} q^{\left(s-  \displaystyle{\max_{\substack{(p,r)\in\mathcal{Q}(G)\\p=q}}}  \left(\theta_{p,r}-|r-s|^+\right)^+\right)k w_{q,s}}
\end{align*}
\end{lemma}
\begin{proof}
The proof is provided in Appendix \ref{section:Proof_Joint_Prob_T}
\end{proof}

\begin{lemma}
For $a\in J$ and $u\in G^n$, we have
\begin{align*}
P\left(\phi(a)+B=u\right)=\frac{1}{|G|^n}
\end{align*}
\end{lemma}
\begin{proof}
Immediate from Lemma \ref{lemma:Joint_Prob}.
\end{proof}

\begin{lemma}\label{lemma:Theta_S}
For fixed $w=(w_{q,s})_{(q,s)\in\mathcal{S}(G)}$ and for any $a\in J=\bigoplus_{(q,s)\in\mathcal{S}(G)}\bigoplus_{l=1}^{kw_{q,s}} \mathds{Z}_{q^s}^{(l)}$,
\begin{align*}
\{\theta=(\theta_{p,r})_{(p,r)\in\mathcal{Q}(G)}||T_{\theta}(a)|\ne 0\}=\Theta(w)
\end{align*}
where $\Theta(w)$ is defined in Equation \eqref{eqn:Theta_S}.
\end{lemma}
\begin{proof}
Provided in the Appendix \ref{section:Theta_S}.
\end{proof}

We have
\begin{align*}
\mathds{E}\{\alpha(x)\}&=\sum_{u\in A_{\epsilon}^n(U|x)} \sum_{a\in J} P\left(\phi(a)+B=u\right)\\
&=\frac{|A_{\epsilon}^n(U|x)|\cdot |J|}{|G|^n}
\end{align*}
and
\begin{align*}
\mathds{E}\{\alpha(x)^2\}&=\mathds{E}\left\{\sum_{u,\tilde{u}\in A_{\epsilon}^n(U|x)} \sum_{a,\tilde{a}\in J} \mathds{1}_{\{\phi(a)+B=u,\phi(\tilde{a})+B=\tilde{u}\}} \right\}\\
&=\sum_{u,\tilde{u}\in A_{\epsilon}^n(U|x)} \sum_{a,\tilde{a}\in J} P\left(\{\phi(a)+B=u,\phi(\tilde{a})+B=\tilde{u}\}\right) \\
&=\sum_{\theta\in\Theta(w)} \sum_{a\in J} \sum_{u\in A_{\epsilon}^n(U|x)} \sum_{\tilde{a}\in T_{\theta}(a)} \sum_{\substack{\tilde{u}\in A_{\epsilon}^n(U|x)\\\tilde{u}-u\in H_{\theta}^n}} \frac{1}{|G|^n}\cdot\frac{1}{|H_{\theta}|^n}
\end{align*}
Note that the term corresponding to $\theta=\pmb{0}$ is upper bounded by $\mathds{E}\{\alpha(x)\}^2$. Using Lemma \ref{dinesh_lemma}, we have
\begin{align*}
\left|A_{\epsilon}^n(U|x)\cap \left(u+H_{\theta}^n \right) \right|\le  2^{n[H(U|[U]_{\theta} X)+O(\epsilon)]}
\end{align*}
Therefore,
\begin{align*}
\mbox{\small var}\{\alpha\}&= \mathds{E}\{\alpha(x)^2\}-\mathds{E}\{\alpha(x)\}^2\\
&\le \sum_{\substack{\theta\in\Theta(w)\\\theta\ne \pmb{0}}} |J|\cdot |A_{\epsilon}^n(U|x)|
\prod_{(q,s)\in\mathcal{S}(G)} q^{\left(s- \displaystyle{\max_{\substack{(p,r)\in\mathcal{Q}(G)\\p=q}}}
\left(\theta_{p,r}-|r-s|^+\right)^+\right)k w_{q,s}}  \frac{2^{n[H(U|[U]_{\theta} X)+O(\epsilon)]}}{|G|^n\cdot |H_{\theta}|^n}
\end{align*}
Therefore,
\begin{align*}
P\left(\alpha(x)=0\right)&\le \frac{\mbox{\small var}\{\alpha(x)\}}{\mathds{E}\{\alpha(x)\}^2}\\
&\le \sum_{\substack{\theta\in\Theta(w)\\\theta\ne \pmb{0}}}
\prod_{(q,s)\in\mathcal{S}(G)} \!\!\!\!\! q^{\left( \! s- \!\!\!\!\!\! \displaystyle{\max_{\substack{(p,r)\in\mathcal{Q}(G)\\p=q}}} \!\!\!
\left(\theta_{p,r}-|r-s|^+\right)\!^+\!\right) k w_{q,s}} \frac{ 2^{-n[ H(U|X)-H(U|[U]_{\theta} X)-O(\epsilon)]}|G|^n }{|J|\cdot |H_{\theta}|^n}
\end{align*}
Note that $H(U|X)-H(U|[U]_{\theta} X)=H([U]_{\theta}|X)$ and
\begin{align*}
&|J|=\prod_{(q,s)\in\mathcal{S}(G)} q^{ksw_{q,s}}
\end{align*}
Therefore,
\begin{align*}
P&\left(\alpha(x)=0\right)\le\\
 &\sum_{\substack{\theta\in\Theta(w)\\\theta\ne \pmb{0}}} \!\!\!\!\!\exp_2 \!\! \left\{ \!\! -n \!\! \left[\! H([U]_{\theta}|X) \!-\! \log|G:H_{\theta}|
+ \frac{k}{n} \!\!\! \sum_{(q,s)\in\mathcal{S}(G)} \!\!\!\!\!\! w_{q,s} \log q  \displaystyle{\max_{\substack{(p,r)\in\mathcal{Q}(G)\\p=q}}}
\left(\theta_{p,r}-|r-s|^+\right)^+-O(\epsilon)\right]\right\}
\end{align*}
In order for the probability of error to go to zero as $n$ increases, we require the exponent of all the terms to be negative; or equivalently, for $\theta\in\Theta(w)$ and $\theta\ne \pmb{0}$,
\begin{align*}
R\frac{\sum_{(q,s)\in\mathcal{S}(G)}  \displaystyle{\max_{\substack{(p,r)\in\mathcal{Q}(G)\\p=q}}}  \left(\theta_{p,r}-|r-s|^+\right)^+ w_{q,s} \log q} {\sum_{(q,s)\in\mathcal{S}(G)} s w_{q,s} \log q} > \log |G:H_{\theta}|- H([U]_{\theta}|X)
\end{align*}
Therefore, the achievability condition is
\begin{align*}
R>\frac{1}{\omega_{\theta}} \left( \log |G:H_{\theta}|-H([U]_{\theta}|X)\right)
\end{align*}
with the convention $\frac{1}{0}=\infty$ and where
\begin{align*}
&\omega_{\theta}= \frac{\sum_{(q,s)\in\mathcal{S}}  \displaystyle{\max_{\substack{(p,r)\in\mathcal{Q}(G)\\p=q}}}  \left(\theta_{p,r}-|r-s|^+\right)^+   w_{q,s} \log q} {\sum_{(q,s)\in\mathcal{S}} s w_{q,s} \log q}
\end{align*}

Therefore, the achievable rate is equal to
\begin{align*}
R=\min_{\substack{w_{q,s}, (q,s)\in\mathcal{S}(G)\\\sum w_{q,s}=1}} \max_{\substack{\theta\in\Theta(w)\\\theta\ne \pmb{0}}} \frac{1}{\omega_{\theta}} I([U]_{\theta};X)
\end{align*}
%where for the first minimization $S=\bigotimes_{q\in\mathcal{P}(G)} \mathcal{S}_q$ for some $\mathcal{S}_q\subseteq \{1,\cdots,r_q\}$ and the second minimization is over weights $w_{q,s}$ for $q\in\mathcal{P}(G),s\in\mathcal{S}_q$ adding up to one.

%%%%%%%%%%%%%%%%%%%%%%%%%%%%%%%%%%%%%%%%%%%%%%%%%%%%%%%%%%%%%%%%%%%%%%%%%%%%%%%%%%%%%%%%%

\subsection{Simplification of the Rate for the $\mathds{Z}_{p^r}$ Case}\label{section:simplification}
In this section, we provide a proof of Corollary \ref{cor1} by showing that when $G=\mathds{Z}_{p^r}$ for some prime $p$ and positive integer $r$, then $I_{s.c.}^G(U;X)=R_1$ where
\begin{align}\label{eqn:Rate1_Zpr_Source}
R_1=\max_{\theta=1}^{r} \frac{r}{\theta}I([U]_{\theta};X)
\end{align}
% First, we show that the achievable rate is equivalent to
% \begin{align}\label{eqn:Rate2_Zpr_Source}
% R_2=\min_{\tilde{r}=1}^{r} \max_{\theta=1}^{\tilde{r}} \frac{\tilde{r}}{\theta+\tilde{r}-r}I([U]_{\theta};X)
% \end{align}
When $G=\mathds{Z}_{p^r}$ for some prime $p$ and positive integer $r$, we have $\mathcal{S}(G)=\mathcal{S}_p(G)=\{1,2,\cdots,r\}$. For fixed weights $w_s,s\in\mathcal{S}(G)$ adding up to one, define $\tilde{r}=\max \{s\in\mathcal{S}(G)|w_s\ne 0\}$. We have
\begin{align*}
J=\bigoplus_{s=1}^{\tilde{r}}\bigoplus_{l=1}^{kw_s} \mathds{Z}_{p^s}^{(l)}
\end{align*}
For $a\in J$ and for $\theta=1,\cdots, r$, let $T_{\theta}(a)$ be defined as in Lemma \ref{lemma:Joint_Prob}; i.e.
\begin{align*}
T_{\theta}(a)&=\{\tilde{a}\in J| \min_{\substack{s=1\in\{1,\cdots,\tilde{r}\}\\l=1,\cdots, kw_s}} r-s+\pmb{\hat{\theta}}_{s,l}(a,\tilde{a})=\theta \}\
\end{align*}
where for $a,\tilde{a}\in J$ and for $s=1,\cdots,\tilde{r}$, $l=1\cdots,kw_s$, $\pmb{\hat{\theta}}_{s,l}(a,\tilde{a})=\min\{0\le \hat{\theta}_{s,l}\le s|\tilde{a}_{s,l}-a_{s,l}\in p^{\hat{\theta}_{s,l}}\mathds{Z}_{p^s}\}$. Note that for $\theta=r$, we have $T_{\theta}(a)=\{a\}$ and for $0\le \theta<r$,
\begin{align*}
T_{\theta}(a)=\{\tilde{a}\in J|&\forall s=1,\cdots,\tilde{r},l=1,\cdots,kw_s: \tilde{a}_{s,l}-a_{s,l}\in p^{|\theta+s-r|^+}\mathds{Z}_{p^s} \\
 &\exists s,l: \tilde{a}_{s,l}-a_{s,l}\in p^{|\theta+s-r|^+}\mathds{Z}_{p^s} \backslash p^{|\theta+1+s-r|^+}\mathds{Z}_{p^s}\}
\end{align*}
Note that $\left|p^{|\theta+s-r|^+}\mathds{Z}_{p^s}\right|=\min(p^{r-\theta},p^s)$ and $\left|p^{|\theta+1+s-r|^+}\mathds{Z}_{p^s}\right|=\min(p^{r-\theta-1},p^s)$. Therefore,
\begin{align*}%&=\left(\prod_{s,l} \min(p^{r-\theta},p^s)\right)-\left(\prod_{s,l} \min(p^{(r-\theta-1)^+},p^s)\right)\\
|T_{\theta}(a)|
&= \left(\prod_{s=1}^{{r}} \prod_{l=1}^{k w_s} \min(p^{r-\theta},p^s)\right)-\left(\prod_{s=1}^{{r}} \prod_{l=1}^{k w_s} \min(p^{r-\theta-1},p^s)\right)\\
&= \left(\prod_{s=1}^{r-\theta-1} p^{s k w_s}\right) \cdot \left(\prod_{s=r-\theta}^{r} p^{(r-\theta) k w_s}\right) - \left(\prod_{s=1}^{r-\theta-1} p^{s k w_s}\right) \cdot \left(\prod_{s=r-\theta}^{r} p^{(r-\theta-1) k w_s}\right)\\
&=\left(\prod_{s=1}^{r-\theta-1} p^{s k w_s}\right) \cdot \left(\prod_{s=r-\theta}^{r} p^{(r-\theta) k w_s}\right) \cdot \left[1-\frac{1}{p^{k\sum_{s=r-\theta}^r w_s}}\right]
\end{align*}
This means for $\theta< r-\tilde{r}$, $|T_{\theta}|=0$ and for $r-\tilde{r}\le \theta\le r$ and $|T_{\theta}|\ne0$. Therefore $\Theta(w)= \{r-\tilde{r},\cdots,r\}$. The achievable rate is given by
\begin{align}\label{eqn:Rate_Source_Zpr}
R= \min_{\substack{w_1,\cdots,w_r\\w_1+\cdots+w_r=1}} \max_{\substack{\theta\in\Theta(w)\\\theta\ne 0}} \frac{1}{1-\omega_{\theta}} I([U]_{\theta};X)
\end{align}
where for $\theta\in\Theta(w)$,
\begin{align}\label{eqn:omega_Zpr}
1-\omega_{\theta} &= 1- \frac{\sum_{s=1}^r \left(\theta+s-r\right)^+ w_{s}} {\sum_{s=1}^r s w_{s}}
=\frac{\sum_{s=1}^{r-\theta} sw_s+\sum_{s=r-\theta+1}^r (r-\theta)w_s}{\sum_{s=1}^r sw_s}
\end{align}

Note that for $\theta\ge r-\tilde{r}$, we have
\begin{align*}
\frac{r-\theta}{\tilde{r}} \sum_{s=1}^{\tilde{r}} s w_s &=\frac{r-\theta}{\tilde{r}} \sum_{s=1}^{r-\theta} s w_s+ \sum_{s=r-\theta+1}^{\tilde{r}} \frac{s}{\tilde{r}} (r-\theta)w_s\\
&\le \sum_{s=1}^{r-\theta} s w_s+ \sum_{s=r-\theta+1}^{\tilde{r}} (r-\theta)w_s\\
&=(1-\omega_{\theta})\sum_{s=1}^{\tilde{r}} s w_s
\end{align*}
Therefore, it follows that $1-\omega_{\theta}\ge \frac{r-\theta}{\tilde{r}}$ or equivalently, $\omega_{\theta}\le \frac{\theta+\tilde{r}-r}{\tilde{r}}$. Let $\tilde{w}_1=\cdots=\tilde{w}_{\tilde{r}-1}=\tilde{w}_{\tilde{r}+1}=\tilde{w}_r=0$ and let $\tilde{w}_{\tilde{r}}=1$. Define $\tilde{\omega}_{\theta}$ using Equation \eqref{eqn:omega_Zpr} replacing $w$'s with $\tilde{w}$'s to get $\tilde{\omega}_{\theta}=\frac{\theta+\tilde{r}-r}{\tilde{r}}$. It follows that we always have $\omega_{\theta}\le \tilde{\omega}_{\theta}$ and therefore, it is always optimal to put all the weight on $w_{\tilde{r}}$ if we are confined to have $w_{\tilde{r}+1}=\cdots=w_r=0$. It follows that the achievable rate is equivalent to
\begin{align}\label{eqn:Rate_Source_Zpr2}
R=\min_{\tilde{r}\in\{1,\cdots,r\}} \max_{\substack{\theta\in\{r-\tilde{r},\cdots,r\}\\\theta\ne 0}} \frac{\tilde{r}}{\theta+\tilde{r}-r} I([U]_{\theta};X)
\end{align}
For $\tilde{r}< r$, since by convention $\frac{1}{0}=\infty$, the corresponding term is infinity. It implies that $\tilde{r}=r$ achieves the optimal rate. Hence,
\begin{align}\label{eqn:Rate_Source_Zpr2}
I_{s.c.}^G=\max_{\substack{\theta\in\{1,\cdots,r\}}} \frac{r}{\theta} I([U]_{\theta};X)
\end{align}

%%%%%%%%%%%%%%%%%%%%%%%%%%%%%%%%%%%%%%%%%%%%%%%%%%%

\section{Proof of Channel Coding}\label{section:proof_channel}

\subsection{Encoding and Decoding}

Following the analysis of Section \ref{Image_Ensemble}, we construct the ensemble of group codes of length $n$ over $G$ as the image of all homomorphisms $\phi$ from some Abelian group $J$ into $G^n$ where $J$ and $G^n$ are as in Equations \eqref{eqn:J2} and \eqref{eqn:Gn} respectively. The random homomorphism $\phi$ is described in Equation \eqref{eqn:phi}.

To find an achievable rate, we use a random coding argument in which the random encoder is characterized by the random homomorphism $\phi$ and a random vector $B$ uniformly distributed over $G^n$. Given a message $u\in J$, the encoder maps it to $x=\phi(u)+B$ and $x$ is then fed to the channel. At the receiver, after receiving the channel output $y\in\mathcal{Y}^n$, the decoder looks for a unique $\tilde{u}\in J$ such that $\phi(\tilde{u})+B$ is jointly typical with $y$ with respect to the distribution $p_X W_{Y|X}$ where $p_X$ is uniform over $G$. If the decoder does not find such $\tilde{u}$ or if such $\tilde{u}$ is not unique, it declares error.\\

\subsection{Error Analysis}\label{section:ErrAnalysis_General_Channel}
Let $u$, $x$ and $y$ be the message, the channel input and the channel output respectively. The error event can be characterized by the union of two events: $E(u)=E_1(u)\cup E_2(u)$ where $E_1(u)$ is the event that $\phi(u)+B$ is not jointly typical with $y$ and $E_2(u)$ is the event that there exists a $\tilde{u}\ne u$ such that $\phi(\tilde{u})+B$ is jointly typical with $y$. We can provide an upper bound on the probability of the error event as $P(E(u))\le P(E_1(u))+P(E_2(u)\cap (E_1(u))^c)$. Using the standard approach, one can show that $P\left(E_1(u)\right)\rightarrow 0$ as $n\rightarrow \infty$. The probability of the error event $E_2(u)\cap (E_1(u))^c$ averaged over all messages can be written as
\begin{align*}
P_{avg}(E_2(u)\cap (E_1(u))^c)=\sum_{u\in J} \frac{1}{|J|}
\sum_{x\in G^n}\mathds{1}_{\{\phi(u)+B=x\}}
\sum_{y\in A_{\epsilon}^n(Y|x)} W^n_{Y|X}(y|x)
\mathds{1}_{\{\exists \tilde{u}\in J: \tilde{u}\ne u, \phi(\tilde{u})+B\in A_{\epsilon}^n(X|y) \}}
\end{align*}
The expected value of this probability over the ensemble is given by $\mathds{E}\{P_{avg}(E_2(u)\cap (E_1(u))^c)\}=P_{err}$ where
\begin{align*}
P_{err}=\sum_{u\in J} \frac{1}{|J|}
\sum_{x\in G^n} \sum_{y\in A_{\epsilon}^n(Y|x)} W^n_{Y|X}(y|x)
P\left(\phi(u)+B=x, \exists \tilde{u}\in J: \tilde{u}\ne u, \phi(\tilde{u})+B\in A_{\epsilon}^n(X|y) \right)
\end{align*}
Using the union bound, we have
\begin{align*}
P_{err}\le \sum_{u\in J} \frac{1}{|J|}
\sum_{x\in G^n} \sum_{y\in A_{\epsilon}^n(Y|x)}
\sum_{\substack{\tilde{u}\in J\\ \tilde{u}\ne u}}
\sum_{\tilde{x}\in A_{\epsilon}^n(X|y)}
W^n_{Y|X}(y|x)
P\left(\phi(u)+B=x, \phi(\tilde{u})+B=\tilde{x} \right)
\end{align*}
Define $\Theta(w)$ as in Equation \eqref{eqn:Theta_S} and for $\theta\in\Theta(w)$ and $u\in J$, define $T_{\theta}(u)$ as in Lemma \ref{lemma:Joint_Prob_T}. It follows that
\begin{align*}
P_{err} &\le \sum_{u\in J} \frac{1}{|J|}
\sum_{x\in G^n} \sum_{y\in A_{\epsilon}^n(Y|x)}
\sum_{\substack{\theta\in\Theta(w)\\\theta\ne \pmb{r}}}
\sum_{\tilde{u}\in T_{\theta}(u)}
\sum_{\substack{\tilde{x}\in A_{\epsilon}^n(X|y)}}
W^n_{Y|X}(y|x)
P\left(\phi(u)+B=x, \phi(\tilde{u})+B=\tilde{x} \right)
\end{align*}
Using Lemmas \ref{lemma:Joint_Prob}, \ref{dinesh_lemma} and \ref{lemma:Joint_Prob_T}, we have
\begin{align*}
P_{err}&\le \sum_{\substack{\theta\in \Theta(w)\\\theta\ne \pmb{r}}}
\sum_{u\in J} \frac{1}{|J|}
\sum_{x\in G^n} \sum_{y\in A_{\epsilon}^n(Y|x)}
\sum_{\tilde{u}\in T_{\theta}(u)}
\sum_{\substack{\tilde{x}\in A_{\epsilon}^n(X|y)\\ \tilde{x}\in x+H_{\theta^n}}}
W^n_{Y|X}(y|x)
\frac{1}{|G|^n}\frac{1}{|H_{\theta}|^n}\\
&\le \sum_{\substack{\theta\in \Theta(w)\\\theta\ne \pmb{r}}}
\sum_{u\in J} \frac{1}{|J|}
\sum_{x\in G^n} \sum_{y\in A_{\epsilon}^n(Y|x)}
\sum_{\tilde{u}\in T_{\theta}(u)} W^n_{Y|X}(y|x)
2^{n[H(X|Y[X]_{\theta})+O(\epsilon)]}
\frac{1}{|G|^n}\frac{1}{|H_{\theta}|^n}\\
&\le \sum_{\substack{\theta\in \Theta(w)\\\theta\ne \pmb{r}}}
\sum_{u\in J} \frac{1}{|J|}
\prod_{(q,s)\in\mathcal{S}(G)} q^{\left(s- \displaystyle{\max_{\substack{(p,r)\in\mathcal{Q}(G)\\p=q}}}
\left(\theta_{p,r}-|r-s|^+\right)^+\right)k w_{q,s}}
2^{n[H(X|Y[X]_{\theta})+O(\epsilon)]}
\frac{1}{|H_{\theta}|^n}
\end{align*}
Equivalently, this can be written as
\begin{align*}
&P_{err}\le\\
&\sum_{\substack{\theta\in \Theta(w)\\\theta\ne \pmb{r}}} \!\!\!\!\!
\exp_2 \!\! \left\{ \!\! -n \!\! \left[ \! \frac{-k}{n} \!\!\!\! \sum_{(q,s)\in\mathcal{S}(G)} \!\!\! \left(\!\!s \!-\!\!\!\!\!
\displaystyle{\max_{\substack{(p,r)\in\mathcal{Q}(G)\\p=q}}}  \!\!
\left(\theta_{p,r}-|r-s|^+\right)^+ \!\!\right) \!\! w_{q,s} \log q \!-\! H(X|Y[X]_{\theta}) \!+\! \log|H_{\theta}| \!-\! O(\!\epsilon)\!\right]\!\!\right\}
\end{align*}
Recall that $R=\frac{k}{n}\sum_{(q,s)\in\mathcal{S}(G)}s w_{q,s} \log q$. In order for the probability of error to go to zero, we require the exponent of all the terms to be negative; or equivalently, for $\theta\in \Theta(w)$ and $\theta\ne \pmb{r}$,
\begin{align*}
R\frac{\sum_{(q,s)\in\mathcal{S}(G)} \left(s- \displaystyle{\max_{\substack{(p,r)\in\mathcal{Q}(G)\\p=q}}}   \left(\theta_{p,r}-|r-s|^+\right)^+\right) w_{q,s} \log q}{\sum_{q,s} s w_{q,s} \log q} &< \log |H_{\theta}|-H(X|Y[X]_{\theta})
\end{align*}
Therefore, the achievability conditions are
\begin{align*}
R\le \frac{1}{1-\omega_{\theta}} I(X;Y|[X]_{\theta})
\end{align*}
for all $\theta\in \Theta(w)$ such that $\theta\ne \pmb{r}$ where $\omega_{\theta}$ is defined in \eqref{eqn:omega_theta}. This means that the following rate is achievable
\begin{align*}
R=\min_{\substack{\theta\in \Theta(w)\\\theta\ne \pmb{r}}} \frac{1}{1-\omega_{\theta}} I(X;Y|[X]_{\theta})
\end{align*}
If we maximize over the choice of $w$, we can conclude that the rate $R=I_{c.c}^G(X;Y)$ is achievable.

%%%%%%%%%%%%%%%%%%%%%%%%%%%%%%%%%%%%%%%%%%%%%%%%%%%%%%%%%%%%%%%%%%%%%%%%%%%%%

\subsection{Simplification of the Rate for the $\mathds{Z}_{p^r}$ Case}\label{section:simplification_channel}
In this section, we provide a proof of Corollary \ref{cor2} by showing that if $G=\mathds{Z}_{p^r}$ for some prime $p$ and a positive integer $r$, then $I_{c.c.}^G(X;Y)=R_1$ where
\begin{align}\label{eqn:Rate1_Zpr_Channel}
R_1=\min_{\theta=0}^{r-1}\frac{r}{r-\theta} I(X;Y|[X]_{\theta})
\end{align}
First, we show that the achievable rate is equivalent to
\begin{align}\label{eqn:Rate2_Zpr_Channel}
R_2=\max_{\tilde{r}=1}^{r} \min_{\theta=r-\tilde{r}}^{r-1} \frac{\tilde{r}}{r-\theta} I(X;Y|[X]_{\theta})
\end{align}

When $G=\mathds{Z}_{p^r}$ for some prime $p$ and positive integer $r$, we have $\mathcal{S}(G)=\mathcal{S}_p(G)=\{1,2,\cdots,r\}$.
For fixed weights $w_s,s\in\mathcal{S}(G)$ adding up to one, define $\tilde{r}=\max \{s\in\mathcal{S}(G)|w_s\ne 0\}$.
Similarly to the source coding case, we can show that for $\theta\in\Theta(w)=\{r-\tilde{r},\cdots,r\}$, we have $1-w_{\theta}\ge \frac{r-\theta}{\tilde{r}}$
and it is always optimal to put all the weight on $w_{\tilde{r}}$ if we are confined to have $w_{\tilde{r}+1}=\cdots=w_r=0$. It follows that the achievable rate provided in Equation \eqref{eqn:Rate2_Zpr_Channel} is equal to $I_{c.c.}^G(X;Y)$. Next, we show that the rate in Equation \eqref{eqn:Rate2_Zpr_Channel} is equal to the rate in Equation \eqref{eqn:Rate1_Zpr_Channel}. We need the following lemma:

\begin{lemma}\label{lemma:IXtheta_Xthetaprime}
Let $\theta$ and $\tilde{\theta}$ be such that $0\le \tilde{\theta}\le \theta\le r$. Then
\begin{align*}
I(X;Y|[X]_{\theta})\le I(X;Y|[X]_{\tilde{\theta}})
\end{align*}
\end{lemma}
\begin{proof}
Note that $[X]_{\theta}$ and $[X]_{\tilde{\theta}}$ are both functions of $X$ and therefore
\begin{align*}
&I(X;Y|[X]_{\theta})=I(X;Y)-I([X]_{\theta};Y)\\
&I(X;Y|[X]_{\tilde{\theta}})=I(X;Y)-I([X]_{\tilde{\theta}};Y)
\end{align*}
Furthermore, note that since $\tilde{\theta}\le \theta$ the Markov chain $[X]_{\tilde{\theta}}\leftrightarrow [X]_{\theta}\leftrightarrow Y$ holds and therefore, $I([X]_{\theta};Y)\ge I([X]_{\tilde{\theta}};Y)$. This proves the lemma.
\end{proof}
Let $\theta^*$ be the minimizer in Equation \eqref{eqn:Rate1_Zpr_Channel}. For $r-\theta^*\le \tilde{r}< r$ we have:
\begin{align*}
\min_{\theta=r-\tilde{r}}^{r-1} \frac{\tilde{r}}{r-\theta} I(X;Y|[X]_{\theta}) &\le \left[ \frac{\tilde{r}}{r-\theta} I(X;Y|[X]_{\theta})\right]_{\theta=\theta^*}\\
%&=\frac{\tilde{r}}{r-\theta^*} I(X;Y|[X]_{\theta^*})\\
&< \frac{r}{r-\theta^*} I(X;Y|[X]_{\theta^*})\\
&= R_1
\end{align*}
For $\tilde{r}< r-\theta^*$ we have:
\begin{align*}
\min_{\theta=r-\tilde{r}}^{r-1} \frac{\tilde{r}}{r-\theta} I(X;Y|[X]_{\theta}) &\le \left[ \frac{\tilde{r}}{r-\theta} I(X;Y|[X]_{\theta})\right]_{\theta=r-\tilde{r}}\\
&= I(X;Y|[X]_{r-\tilde{r}})\\
&\le I(X;Y|[X]_{\theta^*})\\
&\le R_1
\end{align*}
Therefore, it follows that the rate $R_1$ is equivalent to the rate $R_2$ and hence $I_{c.c.}^G(X;Y)=R_1$.

\subsection{$G$-Symmetric Channels}\label{section:simplification_symmetric}
In this section, we provide a proof of corollary \ref{cor:symmetric}. Note that since we take $\mathcal{X}=G$, we can take the action of $G$ on $\mathcal{X}$ to be the group operation. We need to show that for all subgroups $H$ of $G$, $I(X;Y|[X])=C_H$ where $X=X+H$ and $C_H$ is the mutual information between the channel input and the channel output when the input is uniformly distributed over $H$; in other words, $C_H=I(X;Y|[X]=H)$. This in turn follows by showing that for all $g\in G$
\begin{align*}
I(X;Y|[X]=g+H)=I(X;Y|[X]=H)
\end{align*}
This can be shown as follows:
\begin{align*}
I(X;Y|[X]=g+H)&=\sum_{x\in g+H}\sum_{y\in\mathcal{Y}} \frac{1}{|H|} W(y|x)\log \frac{W(y|x)}{P(y)}\\
&=\sum_{\tilde{x}\in H}\sum_{y\in\mathcal{Y}} \frac{1}{|H|} W(y|\tilde{x}+g)\log \frac{W(y|\tilde{x}+g)}{P(y)}\\
&\stackrel{(a)}{=}\sum_{\tilde{x}\in H}\sum_{y\in\mathcal{Y}} \frac{1}{|H|} W(g\cdot y|\tilde{x}+g)\log \frac{W(g\cdot y|\tilde{x}+g)}{P(y)}\\
&\stackrel{(b)}{=}\sum_{\tilde{x}\in H}\sum_{y\in\mathcal{Y}} \frac{1}{|H|} W(y|\tilde{x})\log \frac{W(y|\tilde{x})}{P(y)}\\
&=I(X;Y|[X]=H)
\end{align*}
where $(a)$ follows since the action of $g$ on $\mathcal{Y}$ is a bijection of $\mathcal{Y}$ and $(b)$ follows from the symmetric property of the channel. Using this result, it can be shown that the rate provided in \cite[Equation (33)]{fagnani_abelian} is equal to $I_{c.c.}^G(X;Y)$. The difference in the appearance of the two expressions is due to the fact that in \cite[Equation (33)]{fagnani_abelian} the minimization is carried out over the subgroups of the input group whereas in the expression for $I_{c.c.}^G(X;Y)$ the minimization is carried out over the resulting subgroups of the output group.

\section{Examples}\label{section:Examples}
In this section, we provide a few examples for both the source coding problem as well as the channel coding problem. We show that when the underlying group is a field, the source coding group mutual information and the channel coding group mutual information are both equal to the Shannon mutual information. We also provide several non-field examples for both problems.
\subsection{Examples for Source Coding}
In this section, we find the rate-distortion region for a few examples. First, we consider the case where the underlying group is a field i.e. when $G=\mathds{Z}_p^m$ for some prime $p$ and positive integer $m$.
In this case, we have $\mathcal{P}(G)=\{p\}$, $\mathcal{R}_p(G)=\{1\}$, $M_{p,1}=m$ and $\mathcal{S}=\mathcal{S}_p(G)=\{1\}$. Since the set $\mathcal{S}$ is a singleton, the only choice for the weights is $w=w_{p,1}=1$ for which
\begin{align*}
\Theta(w)=\left\{0,1\right\}
\end{align*}
and for $\theta=1$, we have $w_{\theta}=0$ and $[U]_{\theta}=U$. Hence
\begin{align*}
I_{s.c.}^G=I(U;X)
\end{align*}
This means when the underlying group is a field, the rate is equal to the regular mutual information between $U$ and $X$ when $U$ is a uniform random variable.

Next, we consider the case where the reconstruction alphabet is $\mathds{Z}_{4}$. In this case, we have $p=2$ and $r=2$. Therefore,
\begin{align*}
R&=\max_{\theta=1}^{2}\frac{2}{\theta} I([U]_{\theta};X)\\
&=\max(2I([U]_1;X),I(U;X))
\end{align*}
where $U$ is uniform over $\mathds{Z}_{4}$, $X$ is the source output and $[U]_1=U+2^1\mathds{Z}_4=X+\{0,2\}$ and the joint distribution is such that $\mathds{E}\{d(U,X)\}\le D$. Therefore,
\begin{align*}
2I([U]_1;X)=I(U+\{0,2\};X)+I(U+\{1,3\};X)
\end{align*}
Hence,
\begin{align*}
R&=\max\left(I(U;X),I(U+\{0,2\};X)+I(U+\{1,3\};X)\right)
\end{align*}

Next,  we consider the case where the reconstruction alphabet is $\mathds{Z}_{8}$. For this source, we have $p=2$ and $r=3$. Following a similar argument as above we have:
\begin{align*}
R=\max &\left(I(U;X), \frac{3}{2} I([U]_2;X), 3 I([U]_1;X)\right)
%R=\min &\left(I(X;Y), \frac{3}{8}\left(I(X;Y|X\in\{0,4\})+I(X;Y|X\in\{1,5\})+I(X;Y|X\in\{2,6\})+I(X;Y|X\in\{3,7\})\right), \right.\\ &\left.\frac{3}{2}\left(I(X;Y|X\in\{0,2,4,6\})+I(X;Y|X\in\{1,3,5,7\})\right)\right)
\end{align*}
where $U$ is uniform over $\mathds{Z}_{8}$, $X$ is the source output, $[U]_1=U+\{0,2,4,6\}$ and $[U]_2=U+\{0,4\}$.\\
Similarly, for channels with input $\mathds{Z}_9$, we have $p=3$, $r=2$ and
\begin{align*}
R&=\max\left(I(U;X),2 I([U]_1;X)\right)
%R&=\min\left(I(X;Y),\frac{2}{3}\left(I(X;Y|X\in\{0,3,6\})+I(X;Y|X\in\{1,4,7\})+I(X;Y|X\in\{2,5,8\})\right)\right)
\end{align*}
where $U$ is uniform over $\mathds{Z}_{9}$, $X$ is the source output and $[U]_1=U+\{0,3,6\}$.

Finally, we consider $G=\mathds{Z}_2\times \mathds{Z}_4$. In this case, $\mathcal{P}(G)=\{2\}$, $\mathcal{R}_2(G)=\{1,2\}$, $\mathcal{S}(G)=\mathcal{S}_2(G)=\{1,2\}$, $\pmb{0}=(0,0)$ and $w=(w_{1},w_{2})$ such that $w_{1}+w_{2}=1$. We have three cases for $\Theta(w)$:

\textbf{(1)} If $w_{2}=0$ (and $w_1=1$), we have $\Theta(w)=\{(0,1),(1,2)\}$. For $\theta=(0,1)$ we have $\omega_{\theta}=1$. Since by convention $\frac{1}{0}=\infty$, this implies that this case cannot be optimal.

\textbf{(2)} If $w_{1}=0$ (and $w_2=1$), we have $\Theta(w)=\{(0,0),(1,1),(1,2)\}$. For $\theta=(1,1)$ we have $\omega_{\theta}=\frac{1}{2}$ and for $\theta=(1,2)$ we have $\omega_{\theta}=0$ therefore,
\begin{align*}
R_2&=\max\left(2 I([U]_{\theta=(1,1)};X),I([U]_{\theta=(1,2)};X)\right)\\
&=\max\left(2 I([U]_{\theta=(1,1)};X),I(X;Y)\right)
\end{align*}
since $[U]_{\theta=(1,2)}=U$.

Finally, \textbf{(3)} If $0<w_1<1$ (and $w_2=1-w_1$), we have $\Theta(w)=\{(0,0),(0,1),(1,1),(1,2)\}$. For $\theta=(0,1)$ we have $\omega_{\theta}=\frac{w_1+w_2}{w_1+2 w_2}=\frac{1}{1+ w_2}$, for $\theta=(1,1)$ we have $\omega_{\theta}=\frac{w_2}{1+ w_2}$, and for $\theta=(1,2)$ we have $\omega_{\theta}=0$; therefore,
\begin{align*}
R_3&=\min_{w_1,w_2} \max\left((1+ w_2) I([U]_{\theta=(0,1)};X), \frac{1+ w_2}{w_2} I([U]_{\theta=(1,1)};X),  I([U]_{\theta=(1,2)};X)\right)\\
&=\min_{w_1,w_2} \max\left((1+ w_2) I([U]_{\theta=(0,1)};X), \frac{1+ w_2}{w_2} I([U]_{\theta=(1,1)};X),  I(U;X)\right)
\end{align*}
The maximum of $R_3$ is achieved when
\begin{align*}
(1+ w_2) I([U]_{\theta=(0,1)};X)=\frac{1+ w_2}{w_2} I([U]_{\theta=(1,1)};X)
\end{align*}
or equivalently
\begin{align*}
w_2=\frac{I([U]_{\theta=(1,1)};X)}{I([U]_{\theta=(0,1)};X)}
\end{align*}
Therefore,
\begin{align*}
R_3=\max\left(I([U]_{\theta=(1,1)};X)+I([U]_{\theta=(0,1)};X), I(X;Y)\right)
\end{align*}
Note that similarly to the proof of Lemma \ref{lemma:IXtheta_Xthetaprime} and by noting that $[U]_{\theta=(0,1)}\leftrightarrow [U]_{\theta=(1,1)}\leftrightarrow U \leftrightarrow X$ forms a Markov chain, we can show that
\begin{align*}
I([U]_{\theta=(0,1)};X)\le I([U]_{\theta=(1,1)};X) \le I(X;Y)
\end{align*}
This implies that $R_1\ge R_3$ and $R_2\ge R_3$. Therefore,
\begin{align*}
R=R_3=\max\left(I([U]_{\theta=(1,1)};X)+I([U]_{\theta=(0,1)};X), I(X;Y)\right)
\end{align*}
%A GENERAL EXAMPLE TO BE ADDED LATER.???
%In this section, we find the rate-distortion region for the case where the reconstruction alphabet is $\mathds{Z}_{2}\times \mathds{Z}_{4}$.

\subsection{Examples for Channel Coding}
In this section, we find the achievable rate for a few examples: First, we consider the case where the underlying group is a field i.e. when $G=\mathds{Z}_p^m$ for some prime $p$ and positive integer $m$. As in the source coding case, the only choice for the weights is $w=w_{p,1}=1$ for which $\Theta(w)=\left\{0,1\right\}$. For $\theta=0$, we have $w_{\theta}=1$ and $[U]_{\theta}$ is a trivial random variable. Hence
\begin{align*}
I_{s.c.}^G=I(U;X)
\end{align*}
This means when the underlying group is a field, the rate is equal to the regular mutual information between $U$ and $X$ when $U$ is a uniform random variable.

Next, we consider the case where the channel input alphabet is $\mathds{Z}_{4}$. In this case, we have $p=2$ and $r=2$. Therefore,
\begin{align*}
R&=\min_{\theta=0}^{1}\frac{2}{2-\theta} I(X;Y|[X]_{\theta})\\
&=\min(I(X;Y),2I(X;Y|[X]_1))
\end{align*}
where the channel input $X$ is uniform over $\mathds{Z}_{4}$, $Y$ is the channel output and $[X]_1=X+2^1\mathds{Z}_4=X+\{0,2\}$. Therefore,
\begin{align*}
2I(X;Y|[X]_1)=I(X;Y|X\in\{0,2\})+I(X;Y|X\in\{1,3\})
\end{align*}
Hence,
\begin{align*}
R&=\min\left(I(X;Y),I(X;Y|X\in\{0,2\})+I(X;Y|X\in\{1,3\})\right)
\end{align*}
Next, we consider a channel of input alphabet $\mathds{Z}_8$. For this channel we have $p=2$ and $r=3$. Following a similar argument as above we have:
\begin{align*}
R=\min &\left(I(X;Y), \frac{3}{2} I(X;Y|[X]_1), 3 I(X;Y|[X]_2)\right)
%R=\min &\left(I(X;Y), \frac{3}{8}\left(I(X;Y|X\in\{0,4\})+I(X;Y|X\in\{1,5\})+I(X;Y|X\in\{2,6\})+I(X;Y|X\in\{3,7\})\right), \right.\\ &\left.\frac{3}{2}\left(I(X;Y|X\in\{0,2,4,6\})+I(X;Y|X\in\{1,3,5,7\})\right)\right)
\end{align*}
where the channel input $X$ is uniform over $\mathds{Z}_{8}$, $Y$ is the channel output, $[X]_1=X+\{0,2,4,6\}$ and $[X]_2=X+\{0,4\}$.\\
Similarly, for channels with input $\mathds{Z}_9$, we have $p=3$, $r=2$ and
\begin{align*}
R&=\min\left(I(X;Y),2 I(X;Y|[X]_1)\right)
%R&=\min\left(I(X;Y),\frac{2}{3}\left(I(X;Y|X\in\{0,3,6\})+I(X;Y|X\in\{1,4,7\})+I(X;Y|X\in\{2,5,8\})\right)\right)
\end{align*}
where the channel input $X$ is uniform over $\mathds{Z}_{9}$, $Y$ is the channel output and $[X]_1=X+\{0,3,6\}$.

Finally, we consider $G=\mathds{Z}_2\times \mathds{Z}_4$. In this case, $\mathcal{P}(G)=\{2\}$, $\mathcal{R}_2(G)=\{1,2\}$, $\mathcal{S}(G)=\mathcal{S}_2(G)=\{1,2\}$, $\pmb{r}=(1,2)$ and $w=(w_{1},w_{2})$ such that $w_{1}+w_{2}=1$. We have three cases for $\Theta(w)$:

\textbf{(1)} If $w_{2}=0$ (and $w_1=1$), we have $\Theta(w)=\{(0,1),(1,2)\}$. For $\theta=(0,1)$ we have $\omega_{\theta}=1$; therefore,
\begin{align*}
R_1=\frac{1}{\omega_{\theta}} I(X;Y|[X]_{\theta})=I(X;Y|[X]_{\theta=(0,1)})
\end{align*}

\textbf{(2)} If $w_{1}=0$ (and $w_2=1$), we have $\Theta(w)=\{(0,0),(1,1),(1,2)\}$. For $\theta=(1,1)$ we have $\omega_{\theta}=\frac{1}{2}$ and for $\theta=(0,0)$ we have $\omega_{\theta}=1$; therefore,
\begin{align*}
R_2&=\min\left(2 I(X;Y|[X]_{\theta=(1,1)}),I(X;Y|[X]_{\theta=(0,0)})\right)\\
&=\min\left(2 I(X;Y|[X]_{\theta=(1,1)}),I(X;Y)\right)
\end{align*}
since $[X]_{\theta=(0,0)}$ is a trivial random variable.

Finally, \textbf{(3)} If $0<w_1<1$ (and $w_2=1-w_1$), we have $\Theta(w)=\{(0,0),(0,1),(1,1),(1,2)\}$. For $\theta=(0,0)$ we have $\omega_{\theta}=1$, for $\theta=(0,1)$ we have $\omega_{\theta}=\frac{w_1+w_2}{w_1+2 w_2}=\frac{1}{1+ w_2}$ and for $\theta=(1,1)$ we have $\omega_{\theta}=\frac{w_2}{1+ w_2}$ therefore,
\begin{align*}
R_3&=\max_{w_1,w_2} \min\left(\frac{1+ w_2}{ w_2} I(X;Y|[X]_{\theta=(1,1)}), (1+ w_2) I(X;Y|[X]_{\theta=(0,1)}),  I(X;Y|[X]_{\theta=(0,0)})\right)\\
&=\max_{w_1,w_2} \min\left(\frac{1+ w_2}{ w_2} I(X;Y|[X]_{\theta=(1,1)}), (1+ w_2) I(X;Y|[X]_{\theta=(0,1)}),  I(X;Y)\right)
\end{align*}
The maximum of $R_3$ is achieved when
\begin{align*}
\frac{1+ w_2}{ w_2} I(X;Y|[X]_{\theta=(1,1)})= (1+ w_2) I(X;Y|[X]_{\theta=(0,1)}
\end{align*}
or equivalently
\begin{align*}
w_2=\frac{I(X;Y|[X]_{\theta=(1,1)})}{I(X;Y|[X]_{\theta=(0,1)})}
\end{align*}
Therefore,
\begin{align*}
R_3=\min\left(I(X;Y|[X]_{\theta=(1,1)})+ I(X;Y|[X]_{\theta=(0,1)}), I(X;Y)\right)
\end{align*}
Note that similarly to the proof of Lemma \ref{lemma:IXtheta_Xthetaprime} and by noting that $[X]_{\theta=(0,1)}\leftrightarrow [X]_{\theta=(1,1)}\leftrightarrow X \leftrightarrow Y$ forms a Markov chain, we can show that
\begin{align*}
I(X;Y|[X]_{\theta=(1,1)}) \le I(X;Y|[X]_{\theta=(0,1)}) \le I(X;Y)
\end{align*}
This implies that $R_1\le R_3$ and $R_2\le R_3$. Therefore,
\begin{align*}
R=R_3=\min\left(I(X;Y|[X]_{\theta=(1,1)})+ I(X;Y|[X]_{\theta=(0,1)}), I(X;Y)\right)
\end{align*}

\section{Conclusion}\label{section:Conclusion}
We derived the achievable rate-distortion function using Abelian group codes for arbitrary discrete
memoryless sources. We showed that when the underlying group is a
field, these group codes are linear
codes, and this function is equivalent to the symmetric rate-distortion function i.e. the Shannon
rate-distortion function with the additional constraint that the reconstruction random variable is
uniformly distributed. We showed that when the underlying group is not a field, due to the algebraic
structure of the code, certain subgroups of the group appear in the rate-distortion function and cause
a larger rate for a given distortion level. We derived a similar result for the channel
coding problem; i.e. an  achievable rate using Abelian group codes for arbitrary discrete
memoryless channels. We showed that in the case of linear codes, it simplifies to the symmetric capacity
of the channel i.e. the Shannon capacity with the additional constraint that the channel input
distribution is uniformly distributed. For the case where the underlying group is not a field,
as in the source coding case, we observe that several subgroups of the group appear in the
achievable rate and this causes the rate to be smaller than the
symmetric capacity of the channel in general.

\section{Appendix}

\subsection{Proof of Lemma \ref{lemma:g}}\label{section:Proof_Image_Ensemble}
We first prove that for a homomorphism $\phi$, $g_{(q,s,l)\rightarrow (p,r,m)}$ satisfies the above conditions. First assume $p\ne q$. Note that the only nonzero component of $\mathbb{I}_{J:q,s,l}$ takes values from $\mathds{Z}_{q^s}$ and therefore
\begin{align*}
q^s \mathbb{I}_{J:q,s,l}=\overbrace{\sum}_{i=1,\cdots,q^s}^{(J)} \mathbb{I}_{J:q,s,l}=0
\end{align*}
Note that since $\phi$ is a homomorphism, we have $\phi(q^s \mathbb{I}_{J:q,s,l})=0$. On the other hand,
\begin{align*}
\phi(q^s \mathbb{I}_{J:q,s,l})&=\phi(\overbrace{\sum}_{i=1,\cdots,q^s}^{(J)}  \mathbb{I}_{J:q,s,l})\\
&=\overbrace{\sum}_{i=1,\cdots,q^s}^{(\tilde{G})} \phi(\mathbb{I}_{J:q,s,l})\\
&=\bigoplus_{(p,r,m)\in\mathcal{G}(\tilde{G})} \left[\overbrace{\sum}_{i=1,\cdots,q^s}^{(\tilde{G})} \phi(\mathbb{I}_{J:q,s,l})\right]_{p,r,m}\\
&=\bigoplus_{(p,r,m)\in\mathcal{G}(\tilde{G})} \overbrace{\sum}_{i=1,\cdots,q^s}^{(\mathds{Z}_{p^r})} \left[\phi(\mathbb{I}_{J:q,s,l})\right]_{p,r,m}\\
&=\bigoplus_{(p,r,m)\in\mathcal{G}(\tilde{G})} q^s \left[\phi(\mathbb{I}_{J:q,s,l})\right]_{p,r,m}\\
&=\bigoplus_{(p,r,m)\in\mathcal{G}(\tilde{G})} q^s g_{(q,s,l)\rightarrow (p,r,m)}
\end{align*}
%and this implies,
%\begin{align*}
%\left[\phi(q^s \mathds{I}_{J:q,s,l})\right]_{p,r,m}&=\left[q^s \phi(\mathds{I}_{J:q,s,l})\right]_{p,r,m}\\
%&=q^s g_{(q,s,l)\rightarrow (p,r,m)}
%\end{align*}
Therefore, we have $q^s g_{(q,s,l)\rightarrow (p,r,m)}=0 \pmod{p^r}$ or equivalently $q^s g_{(q,s,l)\rightarrow (p,r,m)}=C p^r$ for some integer $C$. Since $p\ne q$, this implies $p^r|g_{(q,s,l)\rightarrow (p,r,m)}$ and since $g_{(q,s,l)\rightarrow (p,r,m)}$ takes value from $\mathds{Z}_{p^r}$, we have $g_{(q,s,l)\rightarrow (p,r,m)}=0$.\\
%Note that for the case $p=q, r\le s$ the condition above is satisfied by definition of $g_{(q,s,l)\rightarrow (p,r,m)}$.\\
Next, assume $p=q$ and $r\ge s$. Note that same as above, we have $\phi(q^s \mathbb{I}_{J:q,s,l})=0$ and
\begin{align*}
\phi(q^s \mathbb{I}_{J:q,s,l})&=\bigoplus_{(p,r,m)\in\mathcal{G}(\tilde{G})} q^s g_{(q,s,l)\rightarrow (p,r,m)}
\end{align*}
and therefore, $q^s g_{(q,s,l)\rightarrow (p,r,m)}=0 \pmod{p^r}$. Since $g_{(q,s,l)\rightarrow (p,r,m)}$ takes values from $\mathds{Z}_{p^r}$ and $p=q$, this implies $p^{r-s}|g_{(q,s,l)\rightarrow (p,r,m)}$ or equivalently $g_{(q,s,l)\rightarrow (p,r,m)}\in p^{r-s}\mathds{Z}_{p^r}$.\\

Next we show that any mapping described by \eqref{eqn:hom_general_form1} satisfying the conditions of the lemma is a homomorphism. For two elements $a,b\in J$ and for $(p,r,m)\in\mathcal{G}(\tilde{G})$ we have
\begin{align}\label{eqn:ItIsHom1}
\nonumber \left[\phi(a+b)\right]_{p,r,m} &=\left[\phi\left( \bigoplus_{(q,s,l)\in\mathcal{G}(J)} (a_{q,s,l}+_{q^s} b_{q,s,l})\right)\right]_{p,r,m}\\
\nonumber &=\left[\phi\left( \overbrace{\sum}_{(q,s,l)\in\mathcal{G}(J)}^{(J)} (a_{q,s,l}+_{q^s} b_{q,s,l}) \mathbb{I}_{J:q,s,l}\right)\right]_{p,r,m}\\
\nonumber &=\left[\phi\left( \overbrace{\sum}_{(q,s,l)\in\mathcal{G}(J)}^{(J)} \overbrace{\sum}_{i=1,\cdots,a_{q,s,l}+_{q^s} b_{q,s,l}}^{(J)} \mathbb{I}_{J:q,s,l}\right)\right]_{p,r,m}\\
\nonumber &=\left[\overbrace{\sum}_{(q,s,l)\in\mathcal{G}(J)}^{(\tilde{G})} \overbrace{\sum}_{i=1,\cdots,a_{q,s,l}+_{q^s} b_{q,s,l}}^{(\tilde{G})} \phi\left( \mathbb{I}_{J:q,s,l}\right)\right]_{p,r,m}\\
\nonumber &=\overbrace{\sum}_{(q,s,l)\in\mathcal{G}(J)}^{(\mathds{Z}_{p^r})} \overbrace{\sum}_{i=1,\cdots,a_{q,s,l}+_{q^s} b_{q,s,l}}^{(\mathds{Z}_{p^r})} \left[\phi\left( \mathbb{I}_{J:q,s,l}\right)\right]_{p,r,m}\\
&=\overbrace{\sum}_{(q,s,l)\in\mathcal{G}(J)}^{(\mathds{Z}_{p^r})} \overbrace{\sum}_{i=1,\cdots,a_{q,s,l}+_{q^s} b_{q,s,l}}^{(\mathds{Z}_{p^r})} g_{(q,s,l)\rightarrow (p,r,m)}
%&=\overbrace{\sum}_{\substack{(q,s,l)\in\mathcal{G}(J)\\q=p}}^{(\mathds{Z}_{p^r})} \overbrace{\sum}_{i=1,\cdots,a_{q,s,l}+_{q^s} b_{q,s,l}}^{(\mathds{Z}_{p^r})} g_{(q,s,l)\rightarrow (p,r,m)}\\
%
%&=\sum_{\substack{(q,s,l)\in\mathcal{G}(J)\\q=p}} \left(a_{q,s,l}g_{(q,s,l)\rightarrow (p,r,m)}+_{q^s} b_{q,s,l} g_{(q,s,l)\rightarrow (p,r,m)}\right)
\end{align}
%\begin{align*}
%\left[\phi(a+b)\right]_{p,r,m} &=\sum_{(q,s,l)\in\mathcal{G}(J)} (a_{q,s,l}+_{q^s} b_{q,s,l}) g_{(q,s,l)\rightarrow (p,r,m)}\\
%&=\sum_{\substack{(q,s,l)\in\mathcal{G}(J)\\q=p}} \left(a_{q,s,l}g_{(q,s,l)\rightarrow (p,r,m)}+_{q^s} b_{q,s,l} g_{(q,s,l)\rightarrow (p,r,m)}\right)
%\end{align*}
On the other hand, we have
\begin{align}\label{eqn:ItIsHom2}
\nonumber \left[\phi(a)+\phi(b)\right]_{p,r,m}&=\left[\phi(a)\right]_{p,r,m}+_{p^r} \left[\phi(b)\right]_{p,r,m}\\
\nonumber &=\left(\overbrace{\sum}_{(q,s,l)\in\mathcal{G}(J)}^{(\mathds{Z}_{p^r})} a_{q,s,l} g_{(q,s,l)\rightarrow (p,r,m)}\right)+_{p^r} \left(\overbrace{\sum}_{(q,s,l)\in\mathcal{G}(J)}^{(\mathds{Z}_{p^r})} b_{q,s,l} g_{(q,s,l)\rightarrow (p,r,m)}\right)\\
\nonumber &=\left(\overbrace{\sum}_{(q,s,l)\in\mathcal{G}(J)}^{(\mathds{Z}_{p^r})} \overbrace{\sum}_{i=1,\cdots,a_{q,s,l}}^{(\mathds{Z}_{p^r})} g_{(q,s,l)\rightarrow (p,r,m)}\right)+_{p^r} \left(\overbrace{\sum}_{(q,s,l)\in\mathcal{G}(J)}^{(\mathds{Z}_{p^r})} \overbrace{\sum}_{i=1,\cdots,b_{q,s,l}}^{(\mathds{Z}_{p^r})} g_{(q,s,l)\rightarrow (p,r,m)}\right)\\
&=\overbrace{\sum}_{(q,s,l)\in\mathcal{G}(J)}^{(\mathds{Z}_{p^r})} \overbrace{\sum}_{i=1,\cdots,a_{q,s,l}+b_{q,s,l}}^{(\mathds{Z}_{p^r})} g_{(q,s,l)\rightarrow (p,r,m)}
\end{align}
where the addition in $a_{q,s,l}+b_{q,s,l}$ is the integer addition.\\
In order to show that $\phi$ is a homomorphism, it suffices to show that under the conditions of the lemma, Equations \eqref{eqn:ItIsHom1} and \eqref{eqn:ItIsHom2} are equivalent. We show that for a fixed $(q,s,l)\in\mathcal{G}(J)$, if the conditions of the lemma are satisfied, then
\begin{align}\label{eqn:ItIsHom3}
\overbrace{\sum}_{i=1,\cdots,a_{q,s,l}+b_{q,s,l}}^{(\mathds{Z}_{p^r})} g_{(q,s,l)\rightarrow (p,r,m)}= \overbrace{\sum}_{i=1,\cdots,a_{q,s,l}+_{q^s} b_{q,s,l}}^{(\mathds{Z}_{p^r})} g_{(q,s,l)\rightarrow (p,r,m)}
\end{align}
Note that if $p\ne q$, then both summations are zero. Note that we have
\begin{align*}
\overbrace{\sum}_{i=1,\cdots,a_{q,s,l}+b_{q,s,l}}^{(\mathds{Z}_{p^r})} g_{(q,s,l)\rightarrow (p,r,m)}= \overbrace{\sum}_{i=1,\cdots,\left(a_{q,s,l}+b_{q,s,l}\right)\pmod{p^r}}^{(\mathds{Z}_{p^r})} g_{(q,s,l)\rightarrow (p,r,m)}
\end{align*}
and
\begin{align*}
\overbrace{\sum}_{i=1,\cdots,a_{q,s,l}+_{q^s}b_{q,s,l}}^{(\mathds{Z}_{p^r})} g_{(q,s,l)\rightarrow (p,r,m)}= \overbrace{\sum}_{i=1,\cdots,\left(a_{q,s,l}+_{q^s}b_{q,s,l}\right)\pmod{p^r}}^{(\mathds{Z}_{p^r})} g_{(q,s,l)\rightarrow (p,r,m)}
\end{align*}
If $p=q$ and $r\le s$, then we have $\left(a_{q,s,l}+_{q^s} b_{q,s,l}\right) \pmod{p^r}=\left(a_{q,s,l}+ b_{q,s,l}\right) \pmod{p^r}$ and hence it follows that Equation \eqref{eqn:ItIsHom3} is satisfied. If $p=q$ and $r\ge s$, since $g_{(q,s,l)\rightarrow (p,r,m)}\in p^{r-s} \mathds{Z}_{p^r}$ we have
\begin{align*}
\overbrace{\sum}_{i=1,\cdots,a_{q,s,l}+b_{q,s,l}}^{(\mathds{Z}_{p^r})} g_{(q,s,l)\rightarrow (p,r,m)}= \overbrace{\sum}_{i=1,\cdots,\left(a_{q,s,l}+b_{q,s,l}\right)\pmod{p^s}}^{(\mathds{Z}_{p^r})} g_{(q,s,l)\rightarrow (p,r,m)}
\end{align*}
and hence it follows that Equation \eqref{eqn:ItIsHom3} is satisfied.
%Note that if the conditions of the lemma are satisfied then we have
%\begin{align*}
%a_{q,s,l}g_{(q,s,l)\rightarrow (p,r,m)}+_{q^s} b_{q,s,l} g_{(q,s,l)\rightarrow (p,r,m)}&=a_{q,s,l}g_{(q,s,l)\rightarrow (p,r,m)}+_{p^r} b_{q,s,l} g_{(q,s,l)\rightarrow (p,r,m)}
%\end{align*}
%Therefore,
%\begin{align*}
%\left[\phi(a+b)\right]_{p,r,m} &=\sum_{\substack{(q,s,l)\in\mathcal{G}(J)\\q=p}} \left(a_{q,s,l}g_{(q,s,l)\rightarrow (p,r,m)}+_{p^r} b_{q,s,l} g_{(q,s,l)\rightarrow (p,r,m)}\right)\\
%&=\left[\phi(a)\right]_{p,r,m}+\left[\phi(b)\right]_{p,r,m}
%\end{align*}
%Hence, $\phi(a+b)=\phi(a)+\phi(b)$.

\subsection{Proof of Lemma \ref{lemma:Joint_Prob}}\label{section:Proof_Joint_Prob}
Note that since $g_{(q,s,l)\rightarrow (p,r,m)}$'s and $B$ are uniformly distributed, in order to find the desired joint probability, we need to count the number of choices for $g_{(q,s,l)\rightarrow (p,r,m)}$'s and  $B$ such that for $(p,r,m)\in\mathcal{G}(G^n)$,
\begin{align*}
%&\left(\sum_{s=1}^r u_s G_s\right)+B=x\\
%&\left(\sum_{s=1}^r \tilde{u}_s G_s\right)+B=\tilde{x}
&\left(\overbrace{\sum}_{(q,s,l)\in\mathcal{G}(J)}^{(\mathds{Z}_{p^r})} a_{q,s,l} g_{(q,s,l)\rightarrow (p,r,m)}\right)+_{p^r} B_{p,r,m}=u_{p,r,m}\\
&\left(\overbrace{\sum}_{(q,s,l)\in\mathcal{G}(J)}^{(\mathds{Z}_{p^r})} \tilde{a}_{q,s,l} g_{(q,s,l)\rightarrow (p,r,m)}\right)+_{p^r} B_{p,r,m}=\tilde{u}_{p,r,m}
\end{align*}
and divide this number by the total number of choices which is equal to
\begin{align*}
%\prod_{p\in\mathcal{P}(G)} \left(p^r\cdot \prod_{\substack{(q,s,l)\\q=p}} p^s\right)^n=
%\prod_{p\in\mathcal{P}(G)} \prod_{s=1}^{r_p} p^{s(k w_{p,s} n)}\cdot |G|^n
|G|^n\cdot \prod_{(p,r,m)\in\mathcal{G}(G^n)} \prod_{\substack{(q,s,l)\in\mathcal{G}(J)\\q=p}} p^{\min(r,s)} &=|G|^n\cdot \left[\prod_{(p,r,m)\in\mathcal{G}(G)} \prod_{\substack{(q,s,l)\in\mathcal{G}(J)\\q=p}} p^{\min(r,s)}\right]^n\\
\end{align*}
where the term $p^{\min(r,s)}$ appears since the number of choices for $g_{(q,s,l)\rightarrow (p,r,m)}$ is $p^r$ if $p=q, r\le s$ and is equal to $p^s$ if $p=q, r\ge s$.
Since $B$ can take values arbitrarily from $G^n$, the number of choices for the above set of conditions is equal to the number of choices for $g_{(q,s,l)\rightarrow (p,r,m)}$'s such that,
\begin{align*}
&\left(\overbrace{\sum}_{(q,s,l)\in\mathcal{G}(J)}^{(\mathds{Z}_{p^r})} (\tilde{a}_{q,s,l} - a_{q,s,l}) g_{(q,s,l)\rightarrow (p,r,m)}\right)=\tilde{u}_{p,r,m} -u_{p,r,m}
\end{align*}
Note that for all $(q,s,l)\in\mathcal{G}(J)$, $(\tilde{a}_{q,s,l} - a_{q,s,l})g_{(q,s,l)\rightarrow (p,r,m)}\in p^{\theta_{p,r}}\mathds{Z}_{p^r}$. Therefore we require $\tilde{u}_{p,r,m} - u_{p,r,m}\in p^{\theta_{p,r}}\mathds{Z}_{p^r}$ and therefore we require $\tilde{u}-u\in H_{\theta}^n$ or otherwise the probability would be zero.\\
For fixed $p\in\mathcal{P}(G)$ and $r\in\mathcal{R}_p(G)$, let $(q^*,s^*,l^*)\in\mathcal{G}(J)$ be such that $q^*=p$ and
\begin{align*}
\theta_{p,r}&=\left|r-s^*\right|^+ +\hat{\theta}_{q^*,s^*,l^*}
\end{align*}
For fixed $(p,r,m)\in\mathcal{G}(G^n)$, and for $(q,s,l)\ne (q^*,s^*,l^*)$, choose $g_{(q,s,l)\rightarrow (p,r,m)}$ arbitrarily from it's domain. The number of choices for this is equal to
\begin{align*}
\left[\prod_{(p,r,m)\in\mathcal{G}(G)} \prod_{\substack{(q,s,l)\in\mathcal{G}(J)\\q=p\\(q,s,l)\ne (q^*,s^*,l^*)}} p^{\min(r,s)}\right]^n\\
\end{align*}
For each $(p,r,m)\in\mathcal{G}(G^n)$, we need to have
\begin{align*}
&(\tilde{a}_{q^*,s^*,l^*}-a_{q^*,s^*,l^*}) g_{(q^*,s^*,l^*)\rightarrow (p,r,m)}=\tilde{u}_{p,r,m} -u_{p,r,m}-\left(\overbrace{\sum}_{\substack{(q,s,l)\in\mathcal{G}(J)\\(q,s,l)\ne (q^*,s^*,l^*)}}^{(\mathds{Z}_{p^r})} (\tilde{a}_{q,s,l}-a_{q,s,l}) g_{(q,s,l)\rightarrow (p,r,m)}\right)
\end{align*}
Note that the right hand side is included in $p^{\theta_{p,r}}\mathds{Z}_{p^r}$ and $(\tilde{a}_{q^*,s^*,l^*}-a_{q^*,s^*,l^*})$ is included in $p^{\hat{\theta}_{q^*,s^*,l^*}}\mathds{Z}_{(q^*)^{(s^*)}}$. We need to count the number of solutions for $g_{(q^*,s^*,l^*)\rightarrow (p,r,m)}$ in $p^{|r-s^*|^+} \mathds{Z}_{p^r}$. Using Lemma \ref{lemma:axb}, we can show that % the set of solutions is equal to
%\begin{align*}
%\{p^{\hat{\theta}-\hat{\theta}_{q^*,s^*,l^*}}\}\cdots????
%\end{align*}
%Note that this set is contained in $p^{|r-s^*|^+} Z_{p^r}$. Therefore,
the number of solutions is equal to $p^{\hat{\theta}_{q^*,s^*,l^*}}$. The total number of solutions for $\phi$ is equal to
\begin{align*}
\left[\left(\prod_{(p,r,m)\in\mathcal{G}} \prod_{\substack{(q,s,l)\in\mathcal{G}(J)\\q=p\\(q,s,l)\ne (q^*,s^*,l^*)}} p^{\min(r,s)}\right)\cdot p^{\hat{\theta}_{q^*,s^*,l^*}}\right]^n
\end{align*}
Hence we have
\begin{align*}
%&\prod_{(p,r,m)\in\mathcal{G}(G^n)} P\left((u_{q^*,s^*,l^*} -\tilde{u}_{q^*,s^*,l^*}) g_{(q^*,s^*,l^*)\rightarrow (p,r,m)}=\tilde{x}_{p,r,m} -x_{p,r,m}-\left(\overbrace{\sum}_{(q,s,l)\in\mathcal{G}(J)}^{(\mathds{Z}_{p^r})} (u_{q,s,l} -\tilde{u}_{q,s,l}) g_{(q,s,l)\rightarrow (p,r,m)}\right) \right)\\
P\left(\phi(a)+B=u,\phi(\tilde{a})+B=\tilde{u}\right)
& =\frac{\left[\prod_{(p,r,m)\in\mathcal{G}(G)} \left(p^{\hat{\theta}_{q^*,s^*,l^*}} \cdot \prod_{\substack{(q,s,l)\in\mathcal{G}(J)\\q=p\\(q,s,l)\ne (q^*,s^*,l^*)}} p^{\min(r,s)}\right)\right]^n} {\left[\prod_{(p,r,m)\in\mathcal{G}(G)} \prod_{\substack{(q,s,l)\in\mathcal{G}(J)\\q=p}} p^{\min(r,s)}\right]^n}\\
& =\left[\prod_{(p,r,m)\in\mathcal{G}(G)} \prod_{\substack{(q,s,l)\in\mathcal{G}(J)\\q=p\\(q,s,l)= (q^*,s^*,l^*)}} \frac{p^{\hat{\theta}_{q^*,s^*,l^*}}}{p^{\min(r,s)}}\right]^n
\end{align*}
Note that for $(q,s,l)= (q^*,s^*,l^*)$ we have
\begin{align*}
\min(r,s)&= \min(r,s^*)= r-|r-s^*|^+= r-\left(\theta_{p,r}-\hat{\theta}_{q^*,s^*,l^*}\right)
\end{align*}
Therefore, the above probability is equal to
\begin{align*}
\left[\prod_{(p,r,m)\in\mathcal{G}} \prod_{\substack{(q,s,l)\in\mathcal{G}(J)\\q=p\\(q,s,l)= (q^*,s^*,l^*)}} \frac{p^{\hat{\theta}_{q^*,s^*,l^*}}}{p^{ r-\left(\theta_{p,r}-\hat{\theta}_{q^*,s^*,l^*}\right)}}\right]^n
 &=\left[\prod_{(p,r,m)\in\mathcal{G}} \prod_{\substack{(q,s,l)\in\mathcal{G}(J)\\q=p\\(q,s,l)= (q^*,s^*,l^*)}} \frac{1} {p^{ r-\theta_{p,r}}}\right]^n\\
& =\left[\prod_{(p,r,m)\in\mathcal{G}} \frac{p^{\theta_{p,r}}} {p^r}\right]^n
 = \frac{1}{|H_{\theta}|^n}
\end{align*}
%Note that
%\begin{align*}
%|H_{\theta}|=\prod_{(p,r,m)\in \mathcal{G}(G)} p^{r-\theta_{p,r}}
%\end{align*}
%Therefore we have
%\begin{align*}
%&\prod_{(p,r,m)\in\mathcal{G}(G^n)} P\left(\overbrace{\sum}_{(q,s,l)\in\mathcal{G}(J)}^{(\mathds{Z}_{p^r})} (u_{q,s,l} -\tilde{u}_{q,s,l}) g_{(q,s,l)\rightarrow (p,r,m)}=\tilde{x}_{p,r,m} -x_{p,r,m}\right) =\frac{1}{|H_{\theta}|^n}
%\end{align*}
Since the dither $B$ is uniform, we conclude that
\begin{align*}
P\left(\begin{array}{l} \phi(u)+B=x\\\phi(\tilde{u})+B=\tilde{x}\end{array}\right)
%&= \frac{1}{|G|^n} \prod_{(p,r,m)\in\mathcal{G}(G^n)} P\left(\overbrace{\sum}_{(q,s,l)\in\mathcal{G}(J)}^{(\mathds{Z}_{p^r})} (u_{q,s,l} -\tilde{u}_{q,s,l}) g_{(q,s,l)\rightarrow (p,r,m)}=\tilde{x}_{p,r,m} -x_{p,r,m}\right)\\
&=\frac{1}{|G|^n}\frac{1}{|H_{\theta}|^n}
\end{align*}

\subsection{Proof of Lemma \ref{lemma:Joint_Prob_T}}\label{section:Proof_Joint_Prob_T}
Let $\tilde{a}\in T_{\theta}(a)$ be such that for $(q,s,l)\in\mathcal{G}(J)$,
\begin{align*}
\tilde{a}_{q,s,l}-a_{q,s,l}\in q^{\hat{\theta}_{q,s,l}}\mathds{Z}_{q^s}\backslash q^{\hat{\theta}_{q,s,l}+1}\mathds{Z}_{q^s}
\end{align*}
for some $0\le \hat{\theta}_{q,s,l}\le s$. Since for all $\tilde{a}\in T_{\theta}(a)$ and all $(p,r)\in\mathcal{Q}(G)$ \begin{align*}
\min_{\substack{(q,s,l)\in\mathcal{G}(J)\\q=p}} |r-s|^++\hat{\theta}_{q,s,l}
\end{align*}
we require $|r-s|^++\hat{\theta}_{q,s,l}\ge \theta_{p,r}$ or equivalently $\hat{\theta}_{q,s,l}\ge \max_{(p,r)\in\mathcal{Q}(G)} \left(\theta_{p,r}-|r-s|^+\right)^+$ for all $(q,s,l)\in \mathcal{G}(J)$. This means for $(q,s,l)\in\mathcal{G}(J)$, $\tilde{a}_{q,s,l}$ can only take values from
\begin{align*}
a_{q,s,l}+q^{\max_{(p,r)\in\mathcal{Q}(G)}\left(\theta_{p,r}-|r-s|^+\right)^+} \mathds{Z}_{q^s}
\end{align*}
The cardinality of this set is equal to
\begin{align*}
q^{s-\max_{(p,r)\in\mathcal{Q}(G)}\left(\theta_{p,r}-|r-s|^+\right)^+}
\end{align*}
Therefore,
\begin{align*}
|T_{\theta}(a)|\le \prod_{(q,s,l)\in \mathcal{G}(J)} q^{s-\left(\theta_{p,r}-|r-s|^+\right)^+}
\end{align*}

\subsection{Proof of Lemma \ref{section:Theta_S}}\label{section:Theta_S}
For $\theta=(\theta_{p,r})_{(p,r)\in\mathcal{Q}(G)}$, if $|T_{\theta}(a)|\ne 0$, let $\tilde{a}\in T_{\theta}(a)$ such that for $(q,s,l)\in\mathcal{G}(J)$,
\begin{align*}
\tilde{a}_{q,s,l}-a_{q,s,l}\in q^{\hat{\theta}_{q,s,l}}\mathds{Z}_{q^s}\backslash q^{\hat{\theta}_{q,s,l}+1}\mathds{Z}_{q^s}
\end{align*}
for some $0\le \hat{\theta}_{q,s,l}\le s$. For all $(p,r)\in\mathcal{Q}(G)$, $\tilde{a}\in T_{\theta}(a)$ implies
\begin{align*}
\theta_{p,r}=\min_{\substack{(q,s,l)\in\mathcal{G}(J)\\q=p}} |r-s|^++\hat{\theta}_{q,s,l}
\end{align*}
Equivalently since
\begin{align*}
\theta_{p,r}&=\min_{\substack{(q,s)\in\mathcal{S}(G),l=1,\cdots,kw_{q,s}\\q=p}} |r-s|^++\hat{\theta}_{q,s,l}\\
&=\min_{\substack{(q,s)\in\mathcal{S}\\q=p\\w_{q,s}\ne 0}} |r-s|^++\min_{l=1,\cdots,kw_{q,s}}\hat{\theta}_{q,s,l}
\end{align*}
This implies $\theta\in \Theta(w)$. The converse part of the proof is similar and is omitted.

\subsection{Useful Lemmas}

\begin{lemma}\label{lemma:axb}
Let $p$ be a prime and $s,r$ a positive integer such that $s\le r$. For $a\in\mathds{Z}_{p^s}$ and $b\in\mathds{Z}_{p^r}$, let $0\le\hat{\theta}\le s$ and $\hat{\theta}\le \theta\le r$ be such that
\begin{align*}
&a\in p^{\hat{\theta}}\mathds{Z}_{p^s}\backslash p^{\hat{\theta}+1}\mathds{Z}_{p^s}\\
&b\in p^{\theta}\mathds{Z}_{p^r}%\backslash p^{\theta_2+1}\mathds{Z}_{p^r}
\end{align*}
Write $a=p^{\hat{\theta}}\alpha$ for some invertible element $\alpha\in\mathds{Z}_{p^r}$ and $b=p^{\theta}\beta$ for some $\beta\in\beta\in\{0,1,
\cdots,p^{r-\theta}-1\}$. Then, the set of solutions to the equation ${ax}\pmod{p^r} =b$ is
\begin{align*}
\left\{p^{\theta-\hat{\theta}}\alpha^{-1}\beta+i\alpha^{-1}p^{r-\hat{\theta}}|i=0,1,\cdots,p^{\hat{\theta}}-1\right\}
\end{align*}
\end{lemma}
\begin{proof}
Note that the representation of $b$ as $b=p^{\theta}\beta$ is not unique and for any $\tilde{\beta}$ of the form $\tilde{\beta}=\beta+ip^{r-\theta}$ for $i=0,1,\cdots,p^{\theta}-1$, $b$ can be written as $p^{\theta}\tilde{\beta}$. Also, the representation of $a$ as $a=p^{\hat{\theta}}\alpha$ is not unique and for any $\tilde{\alpha}=\alpha+i p^{r-\hat{\theta}}$ for $i=0,1,\cdots,p^{\hat{\theta}}-1$, we have $a=p^{\hat{\theta}} \tilde{\alpha}$. The set of solutions to $ax=b$ is identical to the set of solutions to $p^{\hat{\theta}}x=p^{\theta}\alpha^{-1}\beta$. The set of solutions to the latter is
\begin{align*}
\left\{p^{\theta-\hat{\theta}}\alpha^{-1}\beta+i\alpha^{-1}p^{r-\hat{\theta}}|i=0,1,\cdots,p^{\hat{\theta}}-1\right\}
\end{align*}
It remains to show that this set of solutions is independent of the choice of $\alpha$ and $\beta$. First, we show that the set of solutions is independent of the choice of $\beta$. For $\tilde{\beta}=\beta+jp^{r-\theta}$ for some $j\in\{0,1,\cdots,p^{\theta_2}-1\}$, we have
\begin{align*}
&\left\{p^{\theta-\hat{\theta}}\alpha^{-1}\tilde{\beta}+i\alpha^{-1}p^{r-\hat{\theta}}|i=0,1,\cdots,p^{\hat{\theta}}-1\right\}\\
&\qquad\qquad\qquad\qquad\qquad=\left\{p^{\theta-\hat{\theta}}\alpha^{-1} \left(\beta+jp^{r-\theta}\right)  +i\alpha^{-1}p^{r-\hat{\theta}}|i=0,1,\cdots,p^{\hat{\theta}}-1\right\}\\
&\qquad\qquad\qquad\qquad\qquad=\left\{p^{\theta-\hat{\theta}}\alpha^{-1} \beta +\left(i+j\right)\alpha^{-1}p^{r-\hat{\theta}}|i=0,1,\cdots,p^{\hat{\theta}}-1\right\}\\
&\qquad\qquad\qquad\qquad\qquad\stackrel{(a)}{=}\left\{p^{\theta-\hat{\theta}}\alpha^{-1} \beta +i\alpha^{-1}p^{r-\hat{\theta}}|i=0,1,\cdots,p^{\hat{\theta}}-1\right\}
\end{align*}
where $(a)$ follows since the set $p^{r-\hat{\theta}}\{0,1,\cdots,p^{\hat{\theta}}-1\}$ is a subgroup of $\mathds{Z}_{p^r}$ and $j p^{r-\hat{\theta}}$ lies in this set.

Next, we show that the set of solutions is independent of the choice of $\alpha$. For $\tilde{\alpha}=\alpha+jp^{r-\hat{\theta}}$ for some $j\in\{0,1,\cdots,p^{\hat{\theta}}-1\}$, we have
\begin{align*}
\tilde{\alpha}\left(\alpha^{-1}-\alpha^{-1}jp^{r-\hat{\theta}}\tilde{\alpha}^{-1}\right)=1
\end{align*}

Therefore, it follows that the unique inverse of $\tilde{\alpha}$ satisfies $\alpha^{-1}-\tilde{\alpha}^{-1}\in \alpha^{-1}p^{r-\hat{\theta}}\mathds{Z}_{p^r}$. Assume $\tilde{\alpha}^{-1}=\alpha^{-1}+k \alpha^{-1}p^{r-\hat{\theta}}$. We have,
\begin{align*}
&\left\{p^{\theta-\hat{\theta}}\tilde{\alpha}^{-1}{\beta}+i\tilde{\alpha}^{-1}p^{r-\hat{\theta}} |i=0,1,\cdots,p^{\hat{\theta}}-1\right\}\\
&\qquad\qquad\qquad\qquad=\left\{p^{\theta-\hat{\theta}}\left(\alpha^{-1}+k\alpha^{-1}p^{r-\hat{\theta}}\right)\beta+i\left(\alpha^{-1}+k\alpha^{-1}p^{r-\hat{\theta}}\right)p^{r-\hat{\theta}} |i=0,1,\cdots,p^{\hat{\theta}}-1\right\}\\
&\qquad\qquad\qquad\qquad=\left\{p^{\theta-\hat{\theta}}\alpha^{-1} \beta +\left(i+ikp^{r-\hat{\theta}}+k\beta p^{\theta-\hat{\theta}}\right)\alpha^{-1}p^{r-\hat{\theta}}|i=0,1,\cdots,p^{\hat{\theta}}-1\right\}\\
&\qquad\qquad\qquad\qquad\stackrel{(a)}{=}\left\{p^{\theta-\hat{\theta}}\alpha^{-1} \beta +i\alpha^{-1}p^{r-\hat{\theta}}|i=0,1,\cdots,p^{\hat{\theta}}-1\right\}
\end{align*}
where same as above, $(a)$ follows since the set $p^{r-\hat{\theta}}\{0,1,\cdots,p^{\hat{\theta}}-1\}$ is a subgroup of $\mathds{Z}_{p^r}$ and $(ikp^{r-\hat{\theta}}+k\beta p^{\theta-\hat{\theta}}) p^{r-\hat{\theta}}$ lies in this set.
%Similarly, we can show that the set of solutions is independent of the choice of $\alpha$.
\end{proof}

\begin{lemma}\label{dinesh_lemma}
Let $X$ be a random variable taking values from the group $G$ and for a subgroup $H$ of $G$, define $[X]=X+H$. For $y\in A_{\epsilon}^n(Y)$ and $x\in A_{\epsilon}^n(X|y)$, let $z=[x]=x+H^n$. Then we have
\begin{align*}
\left(x+H^n\right)\cap A_{\epsilon}^n(X|y)= A_{\epsilon}^n(X|zy)
\end{align*}
and
\begin{align*}
(1-\epsilon)2^{n[H(X|Y[X])-O(\epsilon)]}\le \left|\left(x+H^n\right)\cap A_{\epsilon}^n(X|y)\right|\le 2^{n[H(X|Y[X])+O(\epsilon)]}
\end{align*}
\end{lemma}
\begin{proof}
First, we show that $\left(x+H^n\right)\cap A_{\epsilon}^n(X|y)$ is contained in $A_{\epsilon}^n(X|zy)$. Since $z$ is a function of $x$, we have $(x,z,y)\in A_{\epsilon}^n(X,[X],Y)$. For $x'\in \left(x+H^n\right)\cap A_{\epsilon}^n(X|y)$, we have $[x']=x'+H^n=x+H^n=z$ and $(x',z,y)=(x',[x'],y)\in A_{\epsilon}^n(X,[X],Y)$. Therefore, $x'\in A_{\epsilon}^n(X|zy)$ and hence,
\begin{align*}
\left(x+H^n\right)\cap A_{\epsilon}^n(X|y)\subseteq A_{\epsilon}^n(X|zy)
\end{align*}
Conversely, for $x'\in A_{\epsilon}^n(X|zy)$, since $(x,z)\in A_{\epsilon}^n(X,[X])$ where $[X]$ is a function of $X$, we have $[x']=z$. This implies $x'\in z+H^n=x+H^n$. Clearly, we also have $x'\in A_{\epsilon}^n(X|y)$. The claim on the size of the set follows since $(z,y)\in A_{\epsilon}^n([X]Y)$.
\end{proof}%\newpage
\bibliographystyle{plain}
\bibliography{ariabib}
\pagebreak
%\bibliographystyle{IEEEtran}
%\bibliography{IEEEabrv,ariabib}%
\end{document}